\documentclass[aps, prd, twocolumn, english, longbibliography, superscriptaddress, breaklinks=true, showpacs=false,showkeys, nofootinbib]{revtex4}
\usepackage[T1]{fontenc}
\usepackage[latin9]{inputenc}
\usepackage{color}
\usepackage{babel}
\usepackage{amsmath}
\usepackage{amsthm}
\usepackage{amsfonts}
\usepackage{amssymb}

\newtheorem{teo}{Theorem}

\usepackage{subfigure}
\usepackage{graphicx}
\usepackage{physics} 
\newcommand{\sscrst}{\scriptscriptstyle}
 
\usepackage{hyperref}
 
\begin{document}

\title{Rotating charged fluids: Theorems and results for Weyl-type systems}

\author{Marcos L. W. Basso}
\affiliation{Centro de Ci\^encias Naturais e Humanas, Universidade Federal do ABC, Avenida dos Estados 5001, Santo Andr\'e, S\~ao Paulo, 09210-580, Brazil}

\author{Vilson T. Zanchin}
\affiliation{Centro de Ci\^encias Naturais e Humanas, Universidade Federal do ABC, Avenida dos Estados 5001, Santo Andr\'e, S\~ao Paulo, 09210-580, Brazil}

\begin{abstract}

We perform a systematic study of rotating charged fluids and extend several well-known theorems regarding static Weyl-type systems which were recently compiled by Lemos and Zanchin [Phys. Rev. D 80, 024010 (2009)] to rotating and axisymmetric systems. Static Weyl-type systems are composed by static charged fluid configurations obeying the Newton-Maxwell or the Einstein-Maxwell systems of equations in which the electric potential $\phi$ and the timelike metric potential $g_{tt}\equiv - W^ 2$ satisfy the Weyl hypothesis, i.e., $W=W(\phi)$. In the present analysis, both the Newton-Maxwell and Einstein-Maxwell theories that describe nonrelativistic and relativistic systems, respectively, are used to perform a detailed analysis of the general properties of rotating charged fluids rotating charged dust as well as rotating charged fluids with pressure in four-dimensional spacetimes. In comparison to the static (nonrotating) systems, two additional potentials, a metric potential related to rotation and an electromagnetic potential related to the magnetic field, come into play for rotating systems. In each case, constraints between the fluid quantities and the metric and electromagnetic potentials are identified in order to generalize the theorems holding for static charged systems to rotating charged systems. New theorems regarding equilibrium configurations with differential rotation in both the Newtonian and the relativistic theories are stated and proved. For rigidly rotating charged fluids in the Einstein-Maxwell theory, a new ansatz involving the gradient of the metric potentials and the gradient of the electromagnetic potentials is considered in order to prove new theorems. Such an ansatz leads to new constraints between the fluid quantities and the potential fields, so implying new equations of state for the charged fluids. Besides the new results and theorems, several previous results are re-obtained in the present analysis. 
\end{abstract}


\maketitle

\section{Introduction} 
\subsection{Weyl-type systems: Definition and overview}

The task of finding exact solutions to the Einstein field equations coupled with other fields is not simple, however some simplifying assumptions can be made in specific cases. For instance, in the case of the Einstein-Maxwell system of equations, Weyl \cite{Weyl} adopted
a interesting strategy while studying static electric fields in general relativity. Such a strategy consists in assuming a functional relationship between the metric potential $g_{tt} \equiv - W^2$ and the electric potential $\phi$. In the electromagnetic vacuum, Weyl found that such a relationship has to be quadratic in $\phi$, i.e.,  $W^2 = (-\epsilon \phi +\beta)^2 +\gamma$, with $\beta$ and $\gamma$ being arbitrary constants and $\epsilon = \pm 1$. This relation between the metric potential and the electrical potential is known as the Weyl relation, and systems satisfying this relation are known as Weyl-type systems. Beyond the electromagnetic vacuum, Majumdar \cite{Majumdar} and Papapetrou \cite{Papa} considered an electrically charged and pressureless fluid (charged dust) that obeys the Weyl relation in the particular case in which $\gamma = 0$ and found equilibrium configurations that represent matter without pressure, extremal charged, in which the gravitational attraction is balanced by electrical repulsion, resulting in systems that have gravitational mass $m$ equal to the electric charge $q$ in geometrized units. Indeed, these authors showed that, when the relation $W = -\epsilon \phi + \beta$ is valid, and considering zero pressure, the equation of hydrostatic equilibrium implies in $\rho_e = \epsilon \rho_m$, where $\rho_e$ and $\rho_m$ are the charge density and the energy density, respectively. The equality $\rho_e = \epsilon \rho_m$ is usually referred to as the Majumdar-Papapetrou (M-P) condition, while the particular relation $W = - \epsilon \phi + \beta$ is called the M-P relation. Systems obeying the M-P condition and the M-P relation are referred to as Majumdar-Papapetrou-type systems.  Majumdar \cite{Majumdar} also verified that, for a M-P-type system in a static spacetime, the metric can be written in the form
 \begin{equation}\label{eq:metric1}
     ds^2 = -W^2 dt^2 + W^{-2}h_{ij}dx^i\, dx^j,
 \end{equation}
where $W=W(x^i)$ and $h_{ij}=h_{ij}(x^k)$ are functions of the spatial coordinates $x^i\, (i,\,j,\, k=1,\, 2,\, 3)$ only, and $W$ satisfies a Poisson-like equation.  This kind of systems has been studied in some depth by Bonnor and many others \cite{Bonnor:1953ex,das62,bonnor65,deray68,de68,Bonnor:1972wi, 
gautreau,bonnor75,Bonnor:1980nw,Gurses:1998zu,ida,Ivanov:2002jy,varela}. For instance,  Das \cite{das62} showed that if the ratio $\rho_e/\rho_m = \epsilon$ holds, then the relation between the metric and the electric potentials must be the one employed by Majumdar and Papapetrou \cite{Majumdar,Papa}, i.e., $W = -\epsilon \phi + \beta$.  De and Raychaudhuri \cite{deray68} took a step forward by showing that if there is a closed equipotential surface within the charged dust fluid with no singularities, holes, or other kind of matter besides charged dust inside it, then the charged dust fluid corresponds to a Majumdar-Papapetrou system. Gautreau and Hoffman \cite{gautreau} investigated  matter sources with pressure that produce Weyl-type fields in which the Weyl quadratic relation is satisfied. They found that, for a given matter source to satisfy the Weyl relation, it must obey the equilibrium condition $\rho_e\left(\epsilon \phi - \beta  \right) = - \epsilon \left(\rho_m + 3p\right) W$, where $p$ is the pressure of the matter (an isotropic charged fluid). 
On the other hand, inspired by the works in general relativity, Bonnor \cite{Bonnor:1980nw} analyzed equilibrium configurations of charged matter in Newtonian gravitation.   
In the study of Newtonian continuous distributions of charged matter with zero pressure, Bonnor found that equilibrium configurations must obey the Majumdar-Papapetrou condition, i.e., $\rho_e = \epsilon \rho_m$, where now $\rho_m$ stands for the Newtonian mass density and $\rho_e$ is the charge density.   Beyond that, Bonnor showed that the relation between the Newtonian gravitational potential $V$ and the electric potential $\phi$ must be linear, i.e., $V(\phi) = - \epsilon \phi + \beta$. More recently, these results valid for charged dust fluids in Newtonian gravitation and in general relativity, that have been obtained by Bonnor and the many other authors mentioned here, were generalized and extended to
spacetimes with arbitrary number of dimensions ($d\geq 4$) in \cite{Lemos:2009}. 

The next step in the study of static Weyl-type systems in the presence of matter with pressure was made by Guilfoyle \cite{Guilfoyle:1999yb}, who considered the following relation between the metric potential and the electric potential,  $W^2 = \alpha\left(-\epsilon \phi +\beta\right)^2 +\gamma$, where $\alpha$, $\beta$, and $\gamma$ are arbitrary constant parameters. This relation, which is referred to as the Weyl-Guilfoyle relation,  generalizes the usual Weyl relation, mentioned above, in which $\alpha = 1$.  
Based on this hypothesis, 
and on the work by Gautreau and Hoffman \cite{gautreau}, it was possible to establish simple and interesting relations among the field quantities ($W$, $\phi$) and the fluid quantities (energy and charge densities, and pressure). This was also done in Ref.~\cite{Lemos:2009}, where the theorems stated by Guilfoyle and the results obtained by Gautreau and Hoffman were generalized and extended to
spacetimes with arbitrary number of dimensions. 
In addition, a set of exact solutions for charged pressure fluids of Weyl-Guilfoyle-type was also obtained and analyzed by Guilfoyle himself \cite{Guilfoyle:1999yb}. 
Thus, one has that Majumdar-Papapetrou solutions, when glued to exterior electrovacuum solutions, produce the so called Bonnor stars \cite{Bonnor:1972wi,bonnor75,bonnor99}. These stars are composed of fluids that have no pressure and, when sufficiently compact, exhibit the quasiblack hole behavior \cite{Lemos:2003gx,lemoskleberzanchin,Lemos:2006sj,lemoszanchin2008}. Similarly, after making his ansatz, Guilfoyle was able to find very interesting exact solutions for electrically charged pressure fluids with spherical symmetry that, when glued with an exterior Reissner-Nordstr\"om metric, produce neat models for charged stars. In this respect, it was noted in the works of Refs.~\cite{lemoszanchin2010, lemoszanchin2016, lemoszanchin2017,masalemoszanchin2023} that, besides charged pressure stars, the exact solutions found by Guilfoyle \cite{Guilfoyle:1999yb} present a bewildering plethora of diverse types of compact objects such as tension stars, quasiblack holes, regular black holes, and quasinonblack holes.

Regarding rotating electrically charged systems, as far as we know, the seminal works by Israel and Wilson \cite{IsraelWilson} and by Perj\'es \cite{Perjes1971} were the first to test the Weyl strategy for finding exact solutions to the Einstein-Maxwell equations in stationary spacetimes. The presence of rotation, together with electric charge, gives rise to a magnetic field establishing a preferred direction in the spacetime, such a direction corresponding to the rotation axis. Thus, a coordinate system that adapts to this kind of symmetry, i.e., axial symmetry, is usually used.  In particular, Israel and Wilson \cite{IsraelWilson} took the metric in the form 
\begin{equation} \label{eq:metric2}
    ds^2 = -W^2\left(dt + \omega_idx^i\right)^2 + W^{-2}h_{ij}dx^i\, dx^j,
\end{equation}
with $W$, $\omega_i$, and $h_{ij}$ ($i,j=1,\,2,\,3$) being arbitrary functions that do not depend upon the time coordinate $t$. Working on the electrovacuum case, and introducing complex potentials \textit{\`a la} Ernst \cite{Ernst1},  the authors arrived at a Laplace-like equation for a complex scalar function which involves the metric functions $W$ and $\omega_i$ and the electromagnetic potential $A_{\mu}$. In a subsequent work, Israel and Spanos \cite{IsraelSpanos1973} considered the simple case where $h_{ij}$ is the Euclidean three-dimensional metric, and faced the difficulty of the Israel-Wilson-Perj\'es (IWP) metric to describe stationary spacetimes in the presence of matter. To circumvent such a difficulty, they proposed an extension of the IWP metric to an interior region containing a dually charged dust fluid, i.e., a dust fluid carrying electric and magnetic charges. We now know that for the metric \eqref{eq:metric2} to admit a continuous distribution of matter, like a perfect fluid, $h_{ij}$ cannot be the simple Euclidean metric (for a very recent work dealing with this subject, see Ref.~\cite{Gurses:2023csy}).  

Adding matter sources and rotation, the subjects of our main interest here, some studies were done in the quest to build exact solutions of the Einstein-Maxwell equations describing Majumdar-Papapetrou and Weyl-Guilfoyle-type systems. Most of the results found in the literature in this direction are related to dust fluids. For example, Islam \cite{Islam1978} and Bonnor \cite{Bonnor1980b} studied a rotating electrically charged dust fluid with axial symmetry, 
both in Newtonian gravitation and in general relativity, thus obtaining specific solutions by relating the potentials (Newtonian and metric potentials, in each case) with the electromagnetic potentials. A subsequent series of interesting works finding exact solutions for rotating charged dust fluids was published by Islam \cite{Islam1977, Islam1979, Islam1980, Islam1983, Islam1983a, Islam2009} and Chakraborty and Bandyopadhyay \cite{ChakraBand1}. 

Other interesting work is the one by Raychaudhuri \cite{Raycha}, who extended the analysis of Bonnor \cite{Bonnor1980b} by dealing with stationary spacetimes without further assumptions on the symmetry of the spatial part of the  metric $h_{ij}$. Raychaudhuri realized, in particular, that the simplifying hypotheses made by Islam \cite{Islam1978} and Bonnor were equivalent to choosing the relation between the electromagnetic potential $A_{\mu}$ and the velocity of the charged dust $u_{\mu}$ in the form $A_{\mu} = -\frac{1}{a} u_{\mu}$,  with $a$ being a constant parameter. An important result that follows from this ansatz is that, for constant $a$,  the energy density $\rho_m$ and the electric charge density $\rho_e$ of a rotating charged dust fluid obey a simple relationship, namely, $\rho_e  = - a \rho_m$, so generalizing the M-P condition. In fact, Raychaudhuri \cite{Raycha} made the ansatz $A_{\mu} = -\frac{1}{a} u_{\mu}$ but treated the quantity $a$ also as a coordinate-dependent function, in which case the relation $\rho_e= -a\rho_m$ mentioned above does not apply. 

Other simplifying hypothesis that has been made in the study of rotating charged fluids is assuming that the Lorentz force vanishes everywhere.  In other previous works, Som and Raychaudhuri \cite{Som}, as well as Islam \cite{Islam1977}, used such a condition to obtain particular solutions of the Einstein-Maxwell equations corresponding to charged dust distribution in rigid rotation where, in addition, the ratio between the charge and the energy densities is an arbitrary constant, i.e., $\rho_e/\rho_m=$ constant. As well, in the case of vanishing Lorentz force and by using comoving coordinates, Banerjee et al. \cite{Banerjee} found a solution with vanishing electric field and with the time-time component of the metric set to unity. It was also shown in such a work that, once again, the ratio between the charge density and the energy density is of the simple form $\rho_e/\rho_m=$ constant, in agreement with the work by Som and Raychaudhuri \cite{Som}, thereby correcting a previous work by other authors \cite{Misra}. 
Another series of interesting works is due to Islam, Van der Bergh, and Wils \cite{Bergh1983, Bergh1984, Bergh, Wils, Islam1984, Wils1985}, where solutions for axisymmetric differentially rotating charged dust fluids were addressed, as well as a general solution for the case of rigid rotation with a constant charge density to energy density ratio was given.

Extensions of the works by Islam, Bonnor, and Raychaudhuri for the treatment of rotating charged fluids with nonzero pressure were made by Chakarborty and Bandyopadhyaya \cite{ChakraBand2}. Previously, in a series of two interesting papers, Das and Kloster\cite{Das1977,Kloster} studied stationary charged perfect fluids. In particular, in \cite{Kloster}, they 
obtained a class of stationary solutions after making some simplifying assumptions such as constant pressure, zero Lorentz force, and assuming special relationships between the electromagnetic and the metric potentials. Analogously, in \cite{Khater}, exact solutions were obtained by working on the particular case of rigid rotation, both for zero and nonzero Lorentz force.

Taking into account the previous works on rotating charged fluids, in order to better explore the physical properties of such systems, we understand that it is necessary a more in-depth study regarding these rotating and axisymmetric systems. In particular,  the extension of the well-known theorems in the literature that arise from the Weyl hypothesis for nonrotating charged systems, and that were recently compiled in \cite{Lemos:2009}, to include rotation is an important step forward. This is the main goal of the present work.

\subsection{Nomenclature and notation}

In order to ease the reading of the manuscript, following Lemos and Zanchin \cite{Lemos:2009}, we set out here the nomenclature regarding the relations between the metric and potential fields, and also between the fluid quantities that appear throughout the text. The notation used for the several differential operators that appear in the manuscript is also set out here.

\subsubsection{Newton-Maxwell rotating charged fluids}

The rotating Newton-Maxwell systems studied in this work are stationary distributions of charged matter characterized by mass density $\rho_m$,  pressure $p$, angular velocity $\Omega$, and electric charge density $\rho_e$. 
The dynamics of the system is governed by the Poisson, Euler, and Maxwell equations. We consider rotating axisymmetric systems and employ the polar-cylindrical system of coordinates $(r,\, \varphi,\, z)$, besides the universal Newtonian time $t$. The Poisson and the Euler equations, combined, give rise to an effective gravity potential $V$ which includes the Newtonian potential $U$ and the centripetal potential $ -\frac{1}{2}\Omega^2 r^2$, i.e., $V =U-\frac{1}{2}\Omega^2 r^2$. The Maxwell equations introduce an electric potential $\phi$ and a magnetic potential $\psi$, that may be combined into a unique effective electromagnetic potential $\Phi = \phi +\Omega\psi$. The Maxwell equations also introduce an effective charge density due to the coupling between the angular velocity of the fluid and the magnetic field, which is named as Goldreich-Julian density, and is denoted by $\rho_{GJ}$.

First we define the nomenclature related to the potentials for the Newton-Maxwell theory with rotating charged matter.

\begin{itemize}
    \item $V= V(\Phi)$ is the Weyl ansatz for rotating systems in the Newton-Maxwell theory. Systems carrying this hypothesis are called Newtonian Weyl-type systems or simply Weyl systems.
    
    \item  $V(\Phi) = - \epsilon \beta \Phi + \gamma$, where $\epsilon =\pm 1$, and  $\beta\neq 1$ and $\gamma$ are arbitrary constants, is the Weyl-Guilfoyle relation for rotating systems in the Newton-Maxwell theory. Systems carrying this hypothesis are called Newtonian Weyl-Guilfoyle-type systems or simply Weyl-Guilfoyle systems. This relation is interesting for rigidly rotating charged dust.
    
    \item $V(\Phi) = - \epsilon \Phi + \gamma$ is the Majumdar-Papapetrou relation for rotating systems in the Newton-Maxwell theory.  The Majumdar-Papapetrou relation is a particular case of the Weyl-Guilfoyle relation.  Systems carrying this hypothesis are called Newtonian Majumdar-Papapetrou systems or simply Majumdar-Papapetrou systems.  This relation is interesting for rigidly rotating charged dust.
\end{itemize}

Now, let us set out the nomenclature regarding the fluid quantities.

\begin{itemize}
 \item $\rho_e = \epsilon \beta \rho_m$, with $\beta\neq 1$, is the Weyl-Guilfoyle condition for the Newton-Maxwell theory. It holds for rigidly rotating dust fluids that obey the Weyl-Guilfoyle relation between the potentials, and that also obey the constraint $ \left[1 -\beta^2\left( 1 - r^2 \Omega^2\right)\!\right]\rho_e + \beta^2\rho_{\sscrst{GJ}} - \frac{\epsilon \beta \Omega^2}{2\pi} = 0$.

 \item  $\rho_e = \epsilon \rho_m$ is the Majumdar-Papapetrou condition (or M-P constraint) for the Newton-Maxwell theory. It is the  particular case of the Weyl-Gilfoyle condition with $\beta=1$. It holds for rigidly rotating dust fluids that obey the Majumdar-Papapetrou relation
 between the potentials, and that also obey the constraint 
$  \Omega^2 r^2 \rho_e +  \rho_{\sscrst{GJ}} - \frac{\epsilon \Omega^2}{2\pi} = 0$.

\item When there is pressure, the relation between the potentials $V$ and $\Phi$ is of the Weyl-Guilfoyle form, and the charged fluid is in rigid rotation, then the relation between the fluid quantities is more intricate and given by $ \rho_m - \frac{\Omega^2}{2\pi} - \epsilon \beta \rho_e\left(1 -  \Omega^2 r^2\right) + \epsilon \beta\rho_{GJ}= 0$, together with  $\rho_e = \epsilon \beta \rho_m - p'$, where $p'$ is the derivative of the pressure $p$ with respect to $\Phi$. 
    
\end{itemize}

Regarding to notation, for the Newton-Maxwell theory, two differential operators appear in the field equations, $\nabla^{2}$ and $\nabla^{2}_-$. These are differential operators in the $(r,\,z)$ plane and are defined by
$ \nabla^{2} = \partial_r^2 + \frac{1}{r} \partial_r + \partial_z^2$
and 
$ \nabla^{2}_- = \partial_r^2 - \frac{1}{r} \partial_r + \partial_z^2$, respectively,  
 with $\partial_j$ standing for the partial derivative $\partial/\partial x^j$.

\subsubsection{Einstein-Maxwell rotating charged fluids}

The rotating Einstein-Maxwell systems investigated in the present work are stationary distributions of charged matter characterized by energy density $\rho_m$, pressure $p$, four-velocity $u_\mu$, whose spatial part is given in terms of the angular velocity $\Omega$, and charge density $\rho_e$. We consider rotating axisymmetric systems, and employ a coordinate system adapted to the axial symmetry, namely, $x^\mu = (t, \, r,\, z,\, \varphi)$. The spacetime metric is then written in the form
\begin{align}
    ds^2 = -f\, dt^2 + 2k\, dt d\varphi + l\, d \varphi^2 + e^{\mu}\left(dr^2 + dz^2\right), \label{eq:metric3a}
\end{align}
where the metric coefficients depend upon $r$ and $z$ only.  
The dynamics of the system is governed by the Einstein-Maxwell field equations and the equation for the conservation of the energy-momentum tensor. The Einstein field equations, together with the equation for the conservation of the energy-momentum tensor, give rise to an effective metric potential $\mathcal{F}$, also called the redshift factor,
which is the general relativistic analogous of the effective Newtonian potential $V$. The second important metric function appearing in the field equations is related to the Killing rotational direction in space, and is denoted by $K$.
In terms of the metric \eqref{eq:metric3a}, the relevant quantities are $F\equiv\mathcal{F}^2 = f- 2\Omega k- \Omega^2 l$, and $K= k +\Omega l$. 
The Maxwell equations introduce a gauge potential $A_\mu$ that is decomposed into an electric potential $\phi$ and a magnetic potential $\psi$. These potentials are combined into a unique effective electromagnetic potential in the form $\Phi = \phi +\Omega\psi$. The usual electromagnetic energy density associated to the electromagnetic field is denoted by $\rho_{em}$, which  may be conveniently decomposed into the electric and magnetic parts, $\rho_{el}$ and $\rho_{mg}$, respectively, so that it holds the relation $ \rho_{em} = \rho_{el}+\rho_{mg}$.

We first define the nomenclature related to the potentials for the Einstein-Maxwell theory with rotating charged matter. 

\begin{itemize}
    \item $\mathcal{F} = \mathcal{F}(\Phi)$ is the Weyl relation for rotating  charged systems, also known as the  Weyl ansatz. Systems carrying this hypothesis are Weyl-type systems.
\end{itemize}
 
\begin{itemize}
    \item $F = \alpha( -\epsilon \Phi + \gamma)^2 + \beta$, where $\epsilon=\pm 1$ and with $\alpha$, $\beta$, and $\gamma$ being arbitrary constants, is the Weyl-Guilfoyle relation for rotating charged systems. Systems carrying this hypothesis are Weyl-Guilfoyle-type systems.

    \item $F = (- \epsilon \Phi + \gamma)^2$, or $\mathcal{F} = - \epsilon \Phi + \gamma$, is the Majumdar-Papapetrou relation for rotating charged systems, which is a particular case of the Weyl-Guilfoyle relation. Systems carrying this hypothesis are Majumdar-Papapetrou-type systems.
 
    \item $u_{\mu} = \epsilon \sqrt{\alpha} A_{\mu}$, with $\alpha$ being an arbitrary constant, is the Bonnor-Raychaudhuri ansatz. 
    It is equivalent to the relations $\mathcal{F} = - \epsilon \sqrt{\alpha} \Phi$ and $K/\mathcal{F} = \epsilon \sqrt{\alpha} \psi$, and it is equivalent to the Majumdar-Papapetrou relation for nonrotating systems.
    
    \item  $F = \alpha \Phi^2 + \beta$, with $\alpha$ and $\beta$ being arbitrary constants, together with $\frac{dK }{d\psi}= - 2 \alpha \Phi$, is a new ansatz made in the present work. It is equivalent to the relations $\partial_j F = 2 \alpha \Phi \partial_j \Phi$ and $\partial_j K = - 2 \alpha \Phi \partial_j \psi$, which usually differs from the Bonnor-Raychauhuri ansatz. We name it as the Islam ansatz. 

\end{itemize}

Now, let us set out the nomenclature regarding the fluid quantities. 

\begin{itemize}
    \item When there is no pressure, if the Majumdar-Papapetrou relation for the potentials holds, then the Majumdar-Papapetrou condition $\rho_e  = \epsilon \rho_m$ is obeyed provided that a further intricate constraint among the fluid quantities and the potentials is also obeyed. 

    \item When there is no pressure, if the Islam ansatz for the metric and electromagnetic potentials holds, then the fluid quantities obey the relation $ \big[\rho_m + 2\left(1 - \alpha\right) \rho_{el} + 2 \rho_{mg}\big]\mathcal{F} + \alpha \Phi \rho_e = 0$, together with $\rho_e = \epsilon \sqrt{\alpha}\rho_m$. This is the general Islam condition for the dust fluid quantities. 

    \item  $\rho_e = \epsilon \beta  \left(\rho_m +  2 \rho_{mg}\right) $ is the first Islam condition. It holds for rigidly rotating charged dust fluids with vanishing Lorentz force that obey the Islam ansatz.
    
    \item  $\rho_e = \epsilon \sqrt{\alpha} \rho_m$ with $ 2 \rho_{mg}= \left(\alpha-1\right)\left(\rho_m + 2 \rho_{el}\right)$ is the second Islam condition.  It holds for rigidly rotating charged dust fluids with nonvanishing Lorentz force that obey the Islam ansatz.

    \item When there is pressure, if the Weyl-Guilfoyle relation holds, an intricate constraint between the fluid quantities and the metric and electromagnetic potentials follows.

    \item When there is pressure, if the Islam ansatz for metric and electromagnetic potentials holds, then the fluid quantities obey the conditions $ \big[\rho_m + 3 p + 2(1 - \alpha) \rho_{el} + 2 \rho_{mg}\big]\mathcal{F} + \alpha \Phi \rho_e = 0$ and $\mathcal{F}p' + \rho_e = \epsilon \sqrt{\alpha}(\rho_m + p)$, with $p'=\frac{dp}{d\Phi}$. This is the general Islam-Guilfoyle condition.

    \item  $\rho_e = \epsilon \beta  \left(\rho_m + 3p + 2 \rho_{mg}\right) $ with constant $p$ is the first Islam-Guilfoyle condition. It holds for rigidly rotating charged pressure fluids with vanishing Lorentz force that obey the Islam ansatz.
    
    \item  $\rho_m + 3p + 2(1 - \alpha) \rho_{el} + 2 \rho_{mg} -  \epsilon \sqrt{\alpha} \rho_e = 0$ together with $\mathcal{F}p' + \rho_e = \epsilon \sqrt{\alpha}(\rho_m + p)$, where $p'=\frac{dp}{d\Phi}$, is the second Islam-Guilfoyle condition.  It holds for rigidly rotating charged pressure fluids with nonvanishing Lorentz force that obey the Islam ansatz.

\end{itemize}

Regarding notation, for the Einstein-Maxwell theory, two differential operators appear in the field equations, $\nabla^{\dagger 2}$ and $\nabla^{\dagger 2}_-$. These are operators in the $(r,\,z)$ plane and are defined by
    $ \nabla^{\dagger 2} = \partial_r^2 + \frac{\partial_r D}{D} \partial_r + \partial_z^2 + \frac{\partial_z D}{D} \partial_z$  
  and $ \nabla^{\dagger 2}_{-} = \partial_r^2 - \frac{\partial_r D}{D} \partial_r + \partial_z^2 - \frac{\partial_z D}{D} \partial_z$, respectively, where $D^2=fl + k^2$. 
Additionally, the Laplace spatial covariant operator, defined on the hypersurface of constant time coordinate $t$, $\Sigma_t$, and compatible with the induced metric on $\Sigma_t$, is denoted by $\nabla^2$.

\subsection{Structure of the paper}

The remainder of this article is organized as follows. 

Section~\ref{sec:2} is devoted to investigating the general properties of rotating axisymmetric charged fluids in Newtonian physics. We start by formulating the problem and giving the fundamental equations in Sec.~\ref{sec:2a}.
Zero pressure charged fluids, i.e., charged dust fluids, and nonzero pressure charged fluids are studied separately. 
For rotating charged dust fluids, we present some results in the case of systems in differential rotation, and then particularize the analysis to rigidly rotating systems.  Several new theorems are stated and proved, rendering the extension to rotating and axisymmetric dust fluids a set of previous results on static systems by other authors.  This is done in Sec.~\ref{sec:2b}.
The study of rotating charged pressure fluids in Newtonian physics is then presented in Sec.~\ref{sec:2c}, where we follow the same steps as in the subsection for the dust fluids. 

Section~\ref{sec:3} is devoted to investigate the general properties of rotating axisymmetric charged fluids in general relativity.
The general formalism and the fundamental equations of the model are presented in Sec.~\ref{sec:3a}, followed by the study of zero pressure charged fluids, i.e., charged dust fluids, reported in Sec.~\ref{sec:3b}. The general case of the dust fluid in differential rotation is considered first, and then 
we particularize to rigidly rotating charged dust fluids. Several new theorems are stated and proved, rendering the extension to rigidly rotating axisymmetric systems a set of results on static  systems due to  Das \cite{das62}, De and Raychaudhuri \cite{deray68}, and Bonnor \cite{Bonnor:1980nw}. Afterwards, inspired by the works of Bonnor \cite{Bonnor1980b} and Raychaudhuri \cite{Raycha}, a new ansatz that relates the gradient of the metric potentials to the gradient of the electromagnetic potentials is proposed. In reference to the Islam's great contributions to the field, we name it as the Islam ansatz. A new theorem that constrains the fluid quantities and the electromagnetic energy density is then stated. We then discuss some solutions obtained by Islam with vanishing Lorentz force \cite{Islam1977}. The last part of Sec.~\ref{sec:3}, Sec.~\ref{sec:3c}, contains the analysis of rotating charged pressure fluids. After presenting a few results holding for charged fluids in deferential rotation, we particularize the analysis to the case of rigid rotation, Several new theorems are stated and proved, so rendering the extension to rotating axisymmetric systems the results due to Guilfoyle \cite{Guilfoyle:1999yb} and Lemos and Zanchin \cite{Lemos:2009}. Lastly, another new theorem is stated and proved by considering rotating charged fluids obeying the Islam ansatz. 

Finally, in Sec.~\ref{sec:4} we make final remarks and conclude. 

Throughout this work we employ geometric units, in which the gravitational constant $G$ and the speed of light $c$ are set to unity, i.e., $G =1= c $.

\section{Rotating Newton-Maxwell charged fluids with pressure}
\label{sec:2}

\subsection{The model and the basic equations}
\label{sec:2a}

In this section we present the fundamental equations to study rotating fluids with electric charge in Newtonian physics.
Let us first mention that the basic equations for rotating dust fluids with electric charge in Newton-Maxwell theory have been formulated and studied in some depth by Islam \cite{Islam1978} and Bonnor \cite{Bonnor1980b}. 
For comparison, we follow closely these works with the key difference that we also study charged fluid with pressure and investigate the equilibrium configurations of such fluids. 
We consider the dynamics of a stationary Newtonian charged pressure fluid in a three-dimensional Euclidean space ${\mathbb R}^3$ according to the Euler description. Besides, we shall also assume axial symmetry and, later on, rigid rotation, as in \cite{Islam1978, Bonnor1980b}.  The fluid is characterized by five quantities, namely, the mass density $\rho_m$, the pressure $p$, the velocity flow $v^i$, the charge density $\rho_e$, and the (convective) current density $J^i = \rho_e v^i$.
The field quantities are the gravitational potential $U$, the electric field $E^i$, and the magnetic field $B^i$.
Now, since the system is stationary, the fields and the fluid quantities do not depend upon the universal time $t$.  Hence, the electric field $E^i$ may be written in terms of an electric potential $\phi$ through $E_i = - \nabla_i \phi$, with $\nabla_i$ standing for the covariant derivative compatible with metric in the three-dimensional Euclidean space $\mathbb{R}^3$.  
Similarly, the magnetic field $B^i$ may be written in terms of a magnetic vector potential $A^i$ through the relation $B_i = \varepsilon_{ijk} \nabla^j A^k$, with $ \varepsilon_{ijk}$ being the usual three-dimensional Levi-Civita tensor. The Roman indices $i,\,j,\, k, ...$ range from 1 to 3, with repeated indices indicating summation.
As a consequence, the set of field equations for such a Newton-Maxwell system may be written in the form
\begin{align}
    & \nabla^2 U = 4 \pi  \rho_m, \label{eq:gravpot}\\ 
    & \nabla_i  E^i = 4 \pi \rho_e, \label{eq:max1} \\
    & \varepsilon_{ijk} \nabla^j E^k = 0,  \label{eq:max2} \\
    & \nabla_i B^i = 0,  \label{eq:max3} \\ 
    & \varepsilon_{ijk} \nabla^j B^k = 4 \pi J_i,   \label{eq:max4}
\end{align}
where $\nabla^2\equiv \nabla_i\nabla^i$ is the Laplace operator.
It is also worth noticing that we are only taking into account the electric current due to convection.

The Newton-Maxwell fluid satisfies the continuity and the Euler equations, which read
\begin{align}
    &\partial_t \rho_m  + \nabla_i \left(\rho_m v^i\right) = 0, \label{eq:cont} \\ 
    &\left(\partial_t + v^j \nabla_j\right) v_i + \dfrac{\nabla_i p}{\rho_m} = - \nabla_i U + \dfrac{\rho_e}{\rho_m}\left(E_i + \varepsilon_{ijk} v^j B^k\right), \label{eq:euler}
\end{align}
respectively.

Since by assumption the fluid distribution is stationary and axisymmetric, it is natural to adopt polar-cylindrical coordinates $\big(x^1,x^2,x^3\big) = \left(r,\varphi,z\right)$ such that, due to the imposed symmetry, the fields and the fluid quantities depend solely on the coordinates $r$ and $z$. Therefore, without loss of generality, the fluid velocity and the magnetic field can be written as
\begin{align}
    &v^i = \Omega(r,z)\, r \,\delta^{i}_{\ \varphi}, \label{eq:velocity}\\ 
    & B^i = \beta (r,z) \delta^{i}_{\ r} + \gamma(r,z) \delta^{i}_{\ z}, 
\end{align}
respectively, where $\Omega^i = \Omega(r,z) \delta^{i}_{\varphi}$ is the angular velocity of the fluid, with $\delta^{i}_{\ j}$ standing for the Kronecker delta tensor. With these choices, the current density takes the form $J^i = \rho_e \Omega\, r\, \delta^{i}_{\ \varphi}$. Now it follows from Eq.~(\ref{eq:max3}) that the nontrivial components of the magnetic field can be written in terms of a unique scalar function $\psi$, namely,
\begin{align}
     &B^r=\beta (r,z) =  \frac{1}{r} \frac{\partial \psi}{\partial z}, \label{eq:Br} \\ &B^z=\gamma(r,z) = -\frac{1}{r} \frac{\partial \psi}{\partial r}, \label{eq:Bz}
\end{align}
which means that the vector potential is of the form $A_i = - \frac{1}{r}\psi\ \delta_i^{\ \varphi}$. Here, $\psi$ is the magnetic potential, which is also called the stream function, because the magnetic field lines are tangent to the curves of constant $\psi$.

At the end, the set of field equations can be rewritten just in terms of three potential fields and of the fluid quantities \cite{Bonnor1980b}, as follows,
\begin{align}
      & \nabla^2 U = 4 \pi\rho_m, \label{eq:pot1} \\ 
      & \nabla^2 \phi = - 4 \pi \rho_e,  \label{eq:pot2}  \\
      & \nabla^{2}_- \psi =  4 \pi r^2 \rho_e \Omega. \label{eq:pot3}
\end{align}
where $\nabla^{2}$ and $\nabla^{2}_-$ are differential operators in the $(r,\,z)$ plane, defined by
\begin{align}
\nabla^{2} \equiv \partial_r^2 + \frac{1}{r} \partial_r + \partial_z^2, \label{eq:nabla}\\
\nabla^{2}_- \equiv \partial_r^2 - \frac{1}{r} \partial_r + \partial_z^2, \label{nabla**}
\end{align}
 with $\partial_j \equiv \partial/\partial x^j$.

Turning to the continuity equation \eqref{eq:cont}, it is straightforward verifying that it is identically satisfied. This result follows from the hypothesis that the fluid is stationary and axisymmetric, i.e., that $\rho_m$ and $v^i$ do not depend upon time $t$ and that the only nonzero component of the velocity is $v^\varphi$, cf. Eq.~\eqref{eq:velocity}.

In turn, the Euler equation \eqref{eq:euler}, for the rotating charged pressure fluid in stationary equilibrium, may be written as
\begin{align}
   \rho_m \nabla_i V +  \rho_e \mathbf{\nabla}_i \Phi + \left(\rho_m \Omega r^2 - \rho_e\psi\right)  \nabla_i \Omega + \nabla_i p  = 0,   \label{eq:equi1} 
\end{align}
where we have defined the effective gravitational potential $V$ 
and the effective electromagnetic potential $\Phi$ by
    \begin{align}
    & V \equiv U - \frac{1}{2}\Omega^2 r^2, \label{eq:eff-gravpot} \\
    & \Phi \equiv \phi + \Omega \psi. \label{eq:eff-electpot}
    \end{align}
As it is seen, the effective gravitational potential $V$ includes the repulsive centrifugal term $-\frac{1}{2}\Omega^2r^2$, while the effective electromagnetic potential $\Phi$ includes the contribution of the magnetic potential $\psi$ through the term $\Omega\,\psi$.

Another interesting aspect to notice in Eq.~\eqref{eq:equi1} is that the quantity $\rho_m \Omega r^2 - \rho_e\psi $ may be thought of as the angular momentum density of the rotating charged fluid. In fact, the quantity $\rho_m \Omega r^2$ is the angular momentum density of the fluid, while $-\rho_e \psi$ is related to the angular momentum density of the electromagnetic field $l_{em}$ through $l_{em} = -\rho_e \psi +  \frac{1}{4 \pi} \nabla_i \left(\psi E^i\right)$. When integrated over a volume $\mathbb{V}$, the quantity $\frac{1}{4 \pi} \nabla_i \big(\psi E^i\big)$ may be transformed into a surface integral through Gauss theorem and, therefore, it can be neglected when the charges and currents lie within a finite volume that is small compared to the volume of integration. In fact, such a surface term vanishes when the integration is taken to the whole space volume. Hence, $L_{em} = - \int \! \rho_e \psi\,  d \mathbb V$ can be interpreted as the total angular momentum of the electromagnetic field, and $L =\int \left(\rho_m \Omega r^2 - \rho_e\psi\right) d \mathbb V$ as the total angular momentum carried by the rotating charged fluid.

The relevant equations for the problem are Eqs.~(\ref{eq:pot1}), (\ref{eq:pot2}), (\ref{eq:pot3}), and the equilibrium equation given in~(\ref{eq:equi1}). Finally, it is worth noticing that Eq.~(\ref{eq:equi1}) is the generalization of the equilibrium equations obtained in \cite{Islam1978, Bonnor1980b} for rotating charged dust fluids in the Newton-Maxwell theory, here by taking into account the fluid pressure.

\subsection{Zero-pressure rotating charged fluids in the Newton-Maxwell theory}
\label{sec:2b}

\subsubsection{Differentially rotating charged dust fluids in the Newton-Maxwell theory}

Let us begin the analysis by considering rotating charged dust fluids, for which $p = 0$.
The problem of finding exact solutions for differentially rotating charged dust fluids in the Newton-Maxwell theory was investigated by Islam \cite{Islam1978, Islam1979, Islam1980, Islam1983} and by Van der Bergh and Wils \cite{Bergh}. Here we follow a different route and make a general analysis of this kind of systems presenting new results.

For zero pressure, after contracting with $dx^i$, the equilibrium equation (\ref{eq:equi1}) reduces to
\begin{align} 
\rho_m d V + \rho_e d \Phi + \left(\rho_m \Omega r^2 - \rho_e \psi\right) d \Omega = 0.
\label{eq:eulerp0}
\end{align}
From this equation, we can state a theorem involving the general properties of equilibrium configurations of charged dust fluids under differential rotation. 

First, for $\rho_m \Omega r^2 - \rho_e \psi \neq 0$, Eq.~\eqref{eq:eulerp0} implies that the potentials $V$, $\Phi$, and $\Omega$ are functionally related.  This means that we can take, for instance, the effective gravitational potential $V$ as a function of $\Phi$ and $\Omega$, i.e.,  $V= V(\Phi,\Omega)$, from what follows $ (\partial V/ \partial \Phi)_{\Omega}= - \rho_e/ \rho_m $  and $  \left(\partial  V/ \partial \Omega\right)_{\Phi}= -\Omega r^2 + \rho_e \psi/{\rho_m}$.  

Next we show that the quantity $\rho_m \Omega r^2 - \rho_e \psi $ vanishes only in the static case, i.e.,  for $\Omega =0=  \psi$. In fact, the quantity $\rho_m \Omega r^2$ is directly related to the kinetic energy density $u_c \equiv \frac{1}{2} \rho_m \Omega^2 r^2 $, which is non-negative, while $ - \rho_e \psi$ is directly related to the magnetic energy density $u_m = \frac{1}{2}J^i A_i = -\frac{1}{2}\rho_e \Omega \psi$, which is also non-negative. Hence, by adding the two energy density contributions one has the inequality $\left(\rho_m \Omega r^2 - \rho_e\psi\right)\Omega\geq 0$, the equality holding just for $\Omega=0$. Now,  on assuming $\Omega > 0$, the inequality reduces to $\rho_m \Omega r^2 - \rho_e\psi\geq 0$. Conversely, on assuming $\Omega < 0$, it is straightforward verifying that $\rho_m \Omega r^2 - \rho_e\psi\leq 0$. Since $\Omega$ is arbitrary, the two inequalities imply in $\rho_m \Omega r^2 - \rho_e\psi\neq 0$, with the equality holding just in the static case where $\Omega=0$ and also $\psi=0$.
The interpretation of the function $\rho_m \Omega^2 r^2 - \rho_e\Omega \psi$ as an effective nonvanishing energy density of the rotating charged fluid can be found, for instance, in the work by Islam \cite{Islam1983}, where an exact solution satisfying the conditions $\Omega r^2 = \psi$ and  $\rho_m = \rho_e$, so that $\rho_m \Omega r^2 - \rho_e \psi=0$, is reported.  It is shown that such a solution corresponds to a non-physical system since the number density results negative, what is equivalent to a fluid with negative Newtonian mass density. 
In summary,  in this work we shall assume  $\rho_m \Omega r^2 - \rho_e \psi \neq 0$, vanishing only for the static limit, and then we can state a theorem.

\begin{teo}[\textit{new}] \label{teo:2}
For any distribution of a differentially rotating charged dust in equilibrium in the Newton-Maxwell theory, if any two of the surfaces of constant $V$, $\Phi$, or $\Omega$ coincide, then the third also coincides.
\end{teo}

\begin{proof}
First notice that $V$, $\Phi$, and $\Omega$ are scalar fields in the three-dimensional Euclidean space $\mathbb{R}^3$, and therefore the
condition of constant $V$, $\Phi$, or $\Omega$ defines a surface in $\mathbb{R}^3$, on which $dV=0$, $d\Phi=0$, or $d\Omega=0$, respectively. Given that  $ \rho_m d V + \rho_e d \Phi + \left(\rho_m \Omega r^2 - \rho_e \psi\right) d \Omega = 0$ and $\rho_m \Omega r^2 - \rho_e \psi \neq 0$, then if any two of the differentials $dV$, $d\Phi$, or $d\Omega$ vanish, then all of them vanish implying that the three surfaces coincide. 
\end{proof}

\subsubsection{Differentially rotating charged dust fluids of Weyl-type in the Newton-Maxwell theory}

Let us now make the Weyl ansatz $V = V(\Phi)$ and explore some of its consequences for differentially rotating charged dust fluids.  The equilibrium equation then reads $\left(\rho_m V' + \rho_e\right) d \Phi + \left(\rho_m \Omega r^2 - \rho_e \psi\right) d \Omega = 0$, where the prime stands for the derivative with respect to $\Phi$. Therefore,  considering $\rho_m V' + \rho_e \neq 0$ and taking into account that $\rho_m \Omega r^2 - \rho_e \psi\neq 0$,  the angular velocity of the fluid results also a function of $\Phi$ alone, i.e., $\Omega = \Omega(\Phi)$. This allows us to state another theorem.

\begin{teo}[\textit{new}] \label{teo:2a} 

If a differentially rotating charged dust is of Weyl-type and is in equilibrium in the Newton-Maxwell theory, then the equipotential surfaces are also surfaces of constant angular velocity, and vice versa.
\end{teo}

\begin{proof}
Using the Weyl ansatz $V = V(\Phi)$ and the fact that $\left(\rho_m V' + \rho_e\right) d \Phi + \left(\rho_m \Omega r^2 - \rho_e \psi\right) d \Omega = 0$, with $\rho_m V' + \rho_e \neq 0$ and $\rho_m \Omega r^2 - \rho_e \psi \neq 0$, then the surfaces of constant $V$ and $\Phi$ coincide since $d V= V' d \Phi$, and Theorem \ref{teo:2} implies that  $\Omega$ is also a constant on such surfaces. Conversely, a surface of constant $\Omega$ implies $d \Phi = 0$ and thus, by using the relation $d V= V' d \Phi$, it also implies that $d V = 0$, and so the surface of constant $\Omega$ is also a surface of constant $\Phi$ and $V$.
\end{proof}

On the other hand, it is worth noticing that, if the quantity $\rho_m V' + \rho_e$ vanishes, it follows that $d \Omega = 0$ and, therefore, the charged dust fluid in the Newton-Maxwell theory is necessarily in rigid rotation.

\subsubsection{Rigidly rotating charged dust fluids in the Newton-Maxwell theory}

Here we assume that the dust fluid is in rigid rotation, i.e., $d\Omega=0$, and then the equilibrium equation \eqref{eq:eulerp0} reduces to
\begin{align}
    \rho_m d V + \rho_e d \Phi = 0\label{eq:eulerp0rig}
\end{align}
which allows us to state the following theorem.

\begin{teo}[\textit{rigidly rotating and axisymmetric version of Bonnor 1980}]\label{teo:p0rigidN1} 

For any distribution of charged dust in rigid rotation in equilibrium in the Newton-Maxwell theory, the surfaces of constant $V$ coincide with the surfaces of constant $\Phi$, with $V$ and $\Phi$ being functionally related.
\end{teo}

\begin{proof}
The proof is similar to the one given in Theorem \ref{teo:2}. The conditions of constant $V$ and $\Phi$ define two-dimensional spaces (two surfaces) in ${\mathbb R}^3$. Additionally, Eq.~(\ref{eq:eulerp0rig}) implies that if any one of $d V$ or $d\Phi$ is zero, then both of them vanish and the surfaces of constant $V$  and $\Phi$ coincide. Besides, $d V /d\Phi = - \rho_e/ \rho_m$  implying in the relation $V  = V(\Phi)$.
\end{proof}

In other words, Theorem \ref{teo:p0rigidN1} implies that, for rigidly rotating charged dust fluid distributions in the Newton-Maxwell theory, the Weyl ansatz $V= V(\Phi)$ is not a necessary hypothesis. The functional dependency $V= V(\Phi)$ is, in fact, a consequence of the equilibrium condition. This is the rigidly rotating version of a theorem by Bonnor~\cite{Bonnor:1980nw} for static Newton-Maxwell systems, which in turn is the Newtonian version of a theorem stated by De and Raychaudhuri for charged dust fluids in the Einstein-Maxwell theory \cite{deray68}. It is also worth mentioning that, in \cite{Bonnor1980b}, Bonnor investigated equilibrium solutions for rigidly rotating charged dust in general relativity by considering  a priori that $V$ and $\Phi$ are functionally related.

\subsubsection{Rigidly rotating charged dust fluids of Weyl-type in the Newton-Maxwell theory}

Besides the statements made above, cf. Theorem \ref{teo:p0rigidN1}, additional general properties of rigidly rotating charged dust fluids may be obtained through a deeper analysis of the set of field equations given in Eqs.~(\ref{eq:pot1}), (\ref{eq:pot2}), and (\ref{eq:pot3}).  Since $\Omega$ is a constant, after some manipulation those three equations are reduced to just two equations for $V$ and $\Phi$, namely,
\begin{align}
    & \nabla^2 V =  4 \pi  \rho_m - 2 \Omega^2, \label{eq:pot4} \\
    & \nabla^2 \Phi =  - 4\pi \rho_e\big(1 - \Omega^2r^2\big) + \frac{2 \Omega}{r} \partial_r \psi. \label{eq:pot5}
\end{align}
By inspecting the last two equations it is natural to interpret the quantities $\rho_m^{eff} = \rho_m - \frac{\Omega^2}{2\pi} $ and $\rho_e^{eff}=  \rho_e\left(1 - \Omega^2r^2\right) - \frac{ \Omega}{2\pi\,r} \partial_r \psi$ as the effective mass density and the effective charge density, that are the sources for the effective potentials $V$ and $\Phi$, respectively.
By using the fact that $V= V(\Phi)$, Eqs.~\eqref{eq:pot4} and \eqref{eq:pot5} furnish 
\begin{align}
    V'^2 \,\nabla^2\Phi + V'\,V'' \left(\nabla_i \Phi\right)^2 = - 4 \pi \rho_e - 2 \Omega^2 V', \label{eq:field1}
\end{align}
where we also used the equilibrium equation written in the form $\rho_m V' + \rho_e = 0$. Now, given that $- 4\pi \rho_e = \nabla^2 \Phi - \Omega \nabla^2 \psi$, see Eqs.~\eqref{eq:pot2} and \eqref{eq:eff-electpot}, by eliminating the charge density $\rho_e$ in Eq.~(\ref{eq:field1}) we arrive at
\begin{align}
    \big(V'^2 -1\big) \nabla^2 \Phi + V'V'' \big(\nabla_i \Phi\big)^2 + \Omega\big( \nabla^2 \psi + 2 \Omega V'\big) = 0. \label{eq:field2}
\end{align}
When $\Omega = 0$, i.e., for static dust fluids, Eq. (\ref{eq:field2}) reduces to $\left(V'^2 -1\right) \nabla^2 \Phi + V'V'' \big(\nabla_i \Phi\big)^2 = 0$ with $V = U$ and $\Phi = \phi$, a result that was first obtained in \cite{Bonnor:1980nw} and rederived in \cite{Lemos:2009} for spaces with an arbitrary number of dimensions. By defining $\mathcal{Z} \equiv \sqrt{V'^2 - 1}$, Eq. (\ref{eq:field2})  can be recast as
\begin{align}
    \mathcal{Z} \mathbf{\nabla}_i  \big(\mathcal{Z} \mathbf{\nabla}^i\Phi\big) =  - \Omega\big( \nabla^2 \psi + 2 \Omega V'\big). \label{eq:field3}
\end{align}
This allows us to state a new theorem, as follows.

\begin{teo}[\textit{rigidly rotating and axisymmetric Newtonian version of Bonnor 1980}] \label{teo:3} 

Supposing that the right-hand side of
Eq.~(\ref{eq:field3}) vanishes, i.e., if 
\begin{equation}
    \nabla^2 \psi + 2 \Omega V'=0, \label{eq:cond-teo5}
\end{equation} then,

(i) In the Newton-Maxwell theory, if the surfaces of any rigidly rotating charged dust distribution in equilibrium are closed equipotential surfaces and inside these surfaces there are no singularities, holes or alien matter, then the relation between $V$ and $\Phi$ is linear
\begin{equation}
   V = - \epsilon \Phi + \gamma, \label{eq:linear}
\end{equation}
where $\epsilon=\pm 1$ and $\gamma$ is an integration constant, and it follows that the fluid quantities satisfy the Majumdar-Papapetrou condition
\begin{align}
    \rho_e = \epsilon \rho_m. \label{eq:linear1}
\end{align}

(ii) In the Newton-Maxwell theory, if the ratio $\rho_e/\rho_m$ equals a constant $\kappa$, and there are no singularities, holes or alien matter in that region, then the potentials $V$ and $\Phi$ are related by Eq.~(\ref{eq:linear}) with $\kappa = \epsilon$.
\end{teo}

\begin{proof}
We first prove assertion \textit{(i)}. Assuming $\mathcal{Z} \neq 0$ and defining $\mathbf{\nabla}_i \Psi \equiv \mathcal{Z} \mathbf{\nabla}_i \Phi$, since one has $ \nabla^2 \psi + 2 \Omega V' = 0$, it follows that $\nabla^2 \Psi = 0$. Therefore $\Psi$ is a harmonic function, which implies in $ \int_{\mathbb{V}_{\!\sscrst S}} \nabla^2 \Psi d \mathbb{V} = \int_{\sscrst S} \left(\mathbf{\nabla}_i \Psi\right) n^i d S =  0$, with $S$ being the boundary of the finite volume $\mathbb{V}_{\sscrst S}$ in the Euclidean space $\mathbb{R}^3$, $n^i$ being the unit vector normal to $S$, and the Gauss theorem has been used. Now, by integrating the divergence $\nabla_i (\Psi \nabla^i \Psi)$ through the finite volume $\mathbb{V}_{\sscrst S}$, it gives
\begin{align}
    \int_{\mathbb{V}_{\!\sscrst S}}\!\! \nabla_i \left(\Psi \nabla^i \Psi\right) d \mathbb{V} = \int_{\mathbb{V}_{\!\sscrst S}} \left(\nabla_i \Psi\right)^2 d\mathbb{V} = \int_{S} \left(\Psi \nabla_i \Psi\right) n^i dS, \label{eq:teofield2}
\end{align}
where we used Gauss theorem again.  If there exist a closed surface which is an equipotential surface for $\Phi$ and satisfies the other assumptions of assertion {\it (i)}, by identifying such a surface with $S$, it follows that $d \Phi = 0$. Contracting $\nabla_i \Psi =  \mathcal{Z} \nabla_i \phi$ and $dx^i$, it follows that $d \Psi = Z d \Phi = 0$ and so $\Psi$ is constant on $S$. Using this result and Eq.~\eqref{eq:teofield2} one gets
\begin{align}
        \int_{\mathbb{V}_{\! S}} \mathcal{Z}^2 (\nabla_i \Phi)^2 d\mathbb{V} = \Psi \int_{S} (\nabla_i \Psi) n^i dS = 0.
\end{align}
Once the integrand is a positive definite function in which $\nabla_i \Phi \neq 0$, it remains $\mathcal{Z}$ must vanish within the whole volume $\mathbb{V}_{\sscrst S}$, which implies in $V' = - \epsilon$. The straightforward integration of such a relation furnishes Eq.~(\ref{eq:linear}), and finally the equilibrium condition $\rho_m V' + \rho_e = 0$ promptly gives Eq.~(\ref{eq:linear1}).

The proof of assertion \textit{(ii)} is straightforward. Given that $\rho_e / \rho_m = \kappa$, by the equilibrium condition \eqref{eq:eulerp0rig}, it follows that $V' = - \kappa$. Therefore $ \mathcal{Z} \nabla_i (\mathcal{Z} \nabla^i \Phi) = \mathcal{Z}^2 \nabla^2 \Phi = 0$, once $ \nabla^2 \psi + 2 \Omega V' = 0$. Since we are considering a region where there is a charge distribution $\rho_e$, then by virtue of Eq.~(\ref{eq:pot5}) one has $\nabla^2 \Phi \neq 0$, remaining that $\mathcal{Z} = \sqrt{V'^2 - 1} = 0$, which gives $V= -\epsilon \Phi+\,$constant, i.e., Eq.~(\ref{eq:linear}) is recovered so that $V'=-\epsilon$ and then the constant $\kappa$ has to be equal to $\epsilon = \pm 1$, what completes the proof.  
\end{proof}

Theorem \ref{teo:3} generalizes to rotating charged dust fluids a theorem by Bonnor \cite{Bonnor:1980nw} for static charged dust fluids, which is the rigidly rotating Newtonian version of the results for relativistic static fluids found in \cite{Papa,das62, deray68}. Here the charged dust fluid is under rigid rotation. Assertion $(i)$ follows by imposing condition \eqref{eq:cond-teo5}, what turns Eq.~\eqref{eq:field3} equivalent to the case without rotation. On the other hand, the assumption that the ratio $\rho_e/\rho_m$ is a constant implies that the potentials $V$ and $\Phi$ are linearly related, exactly as it happens for nonrotating charged dust fluids. 

The assumed condition $\rho_e/\rho_m=\,$constant is reasonable from the physical grounds because it makes the charge density proportional to the mass density. It is also well motivated from the mathematical point of view since it simplifies the system of equations to be solved and the resulting models may be compared to the many results in the literature for static charged fluids. 

On the other hand, the implications of condition \eqref{eq:cond-teo5} need to be further investigated.
For this task we resort to the studies on magnetohydrodynamics and plasma physics, and make use of the Goldreich-Julian density, which is a charge density arising in the  magnetospheres of rotating stars due to the effects of the electromagnetic fields induced in the medium, see Refs.~\cite{Goldreich, Shapiro}. Such a charge density is given approximately by $
\rho_{\sscrst GJ}= -\frac{\Omega_i B^i}{2 \pi}$, and for the kind of systems we are considering here, i.e., axially symmetric charged dust fluids in rigid rotation, it  reduces to $\rho_{\sscrst GJ}= \frac{\Omega }{2\pi\, r}\,\partial_r \psi$,
where $\psi$ is the magnetic potential, see Eqs.~\eqref{eq:Br} and \eqref{eq:Bz}. Then, by using such a charge density together with Eqs.~\eqref{eq:eff-electpot} and \eqref{eq:pot5}, we find
\begin{align}
\Omega\big( \nabla^2 \psi + 2 \Omega V'\big)& =  \Omega\left(4\pi \rho_e\Omega\,r^2+ \frac{2}{r} \partial_r \psi+ 2 \Omega\,V'\right)\nonumber\\
& =  4\pi \rho_e\Omega^2\,r^2+ 4\pi\,\rho_{GJ}+ 2 \Omega^2\,V'. \label{eq:rightside}
\end{align}
From this result we see that the constraint \eqref{eq:cond-teo5} gives a type of equation of state for the effective charge density. In particular, when the condition is valid and Eq.~(\ref{eq:linear}) holds, one has $V'=\epsilon$ and $V''=0$, so that
\begin{align}
   \rho_{GJ}= \frac{\epsilon \Omega^2}{2\pi } - \rho_e \Omega^2 r^2. \label{eq:rhogj}
\end{align}
Hence, $\rho_{\sscrst GJ}$ can be interpret as an effective charge density due to rotation, with the second term on the right-hand side depending explicitly on the electric charge density of the medium $\rho_e$, while  the first term is present even when $\rho_e$ vanishes locally. 
 
 It is noteworthy that explicit exact solutions in which the right-hand side of Eq.~(\ref{eq:field3}) vanishes, and thus relations~(\ref{eq:linear}), (\ref{eq:linear1}), and the equation of state \eqref{eq:rhogj} are valid, have been found by Islam in \cite{Islam1983}.
Other consequence of the condition $\nabla^2 \psi + 2 \Omega V' = 0$ is that it implies a further restriction over the $z$-component of the magnetic field, which results in the form $ B_z =  -\epsilon \Omega_z + 2 \pi \varepsilon_{zjk}x^j J^k$, where $J^k$ is the current density and $\mathcal{M}_z = \frac{1}{2} \varepsilon_{zjk}x^j J^k$ is the $z$-component of the magnetic moment density. 

Finally, let us analyze a situation where the right-hand side of Eq.~(\ref{eq:field3}) does not vanish, and then the relation between $V$ and $\Phi$ cannot be Eq.~\eqref{eq:linear}.
In such a case, it is interesting to consider a more general linear relation between $V$ and $\Phi$ than~\eqref{eq:linear} and investigate the consequences.  For instance,  Islam \cite{Islam1978} and Bonnor \cite{Bonnor1980b} proposed a simple linear relation of the form  
\begin{align}
V = -\epsilon \beta \Phi + \gamma, \label{eq:lin}    
\end{align}
where $\beta$ and $\gamma$ are arbitrary constant parameters. With this hypothesis, after using also Eqs.~\eqref{eq:pot5} to bring back the  effective charge density,  Eq.~(\ref{eq:field2}) reduces to  
\begin{align}
    \rho_{\sscrst GJ}= \frac{\epsilon \beta \Omega^2}{2\pi} -\left[1 -\beta^2\left( 1 - r^2 \Omega^2\right)\!\right]\!\rho_e.   \label{eq:cond1} 
\end{align}
This relation gives the Goldreich-Julian induced charge density in terms of the fluid quantities, generalizing the relationship \eqref{eq:rhogj} which holds for $\beta =1$.
In turn, the continuity equation yields the Newtonian version of the Weyl-Guilfoyle condition,
\begin{align}
    \rho_e = \epsilon\,\beta\, \rho_m. \label{eq:rhom-rhoe2}
\end{align}
Notice that parameter $\beta$ in relation~(\ref{eq:lin}) produces a similar effect on the rotating dust fluid as the pressure in the static fluid case. In fact, as shown in \cite{Lemos:2009}, a relation in the form~\eqref{eq:rhom-rhoe2} follows directly from the hypothesis~(\ref{eq:lin}) for nonrotating fluids with pressure in the Newton-Coulomb theory. See also \cite{Guilfoyle:1999yb} for the relativistic counterpart of this analysis.

\subsection{Nonzero-pressure rotating charged fluids in the Newton-Maxwell theory}
\label{sec:2c}

\subsubsection{Differentially rotating charged pressure fluids in the Newton-Maxwell theory}

Here we consider differentially rotating charged fluids with nonzero pressure in the Newton-Maxwell theory. The general properties of such kind of systems are investigated by following the previous works on static systems, see e.g. \cite{Bonnor:1980nw,Lemos:2009}.

As in the preceding section, the equilibrium equation (\ref{eq:equi1}) is written in terms of total derivatives by contracting with $d x^i$, i.e.,
\begin{align}
   \rho_md V  +  \rho_e  d\Phi + \left(\rho_m \Omega r^2 - \rho_e  \psi\right)  d \Omega + d p   = 0.  \label{eq:equilp} 
\end{align}
The potentials $V$ and $\Phi$ and the fluid quantities $\Omega$ and $p$ are then functionally related through Eq.~\eqref{eq:equilp}. This means that, for instance, the pressure is a function of the other three quantities, i.e., $p = p(V, \Phi,\Omega)$, what gives  $ (\partial p/ \partial V)_{\Phi,\Omega} = - \rho_m $, $(\partial  p/ \partial \Phi)_{V,\Omega}= - \rho_e$, and $\left(\partial  p/ \partial \Omega\right)_{V,\Phi}= - \left(\rho_m \Omega r^2 - \rho_e  \psi\right)$. Therewith, we can state a new theorem.

\begin{teo}[\textit{differentially rotating and axisymmetric Newtonian version of Guilfoyle 1999}] \label{teo:4dif}
For any differentially rotating charged pressure fluid in equilibrium in the Newton-Maxwell theory, if any three of the four surfaces of constant $V$, $\Phi$, $\Omega$, and $p$ coincide, then the fourth also coincides with the other three.
\end{teo}

\begin{proof} The proof is similar to Theorems \ref{teo:2} and \ref{teo:p0rigidN1}.
The potentials $V$, $\Phi$, and the fluid quantities $\Omega$ and $p$ are scalar functions in the Euclidean space $\mathbb{R}^3$, and then the
conditions of constant $V$, $\Phi$, $\Omega$, or $p$ define level surfaces in $\mathbb{R}^3$ implying in $dV=0$, $d\Phi=0$, $d\Omega=0$, and $dp=0$ on each respective surface. Moreover, since Eq.~(\ref{eq:equilp}) is valid in its full extent, it is straightforward to see that if any three of the four differentials $dV$, $d\Phi$, $d \Omega$, or $d p$ vanish, then all the four of them vanish, meaning that the four surfaces coincide. 
\end{proof}

This theorem is the Newton-Maxwell analogous of the one stated by Guilfoyle \cite{Guilfoyle:1999yb} for relativistic static charged fluids in the Einstein-Maxwell theory. It generalizes Theorem \ref{teo:2} by including the fluid pressure.

\subsubsection{Differentially rotating charged pressure fluids of Weyl-type in the Newton-Maxwell theory}

Let us consider the Weyl ansatz $V = V(\Phi)$ and explore some of its consequences for a differentially rotating charged fluid with pressure. 
The equilibrium equation now reads $(\rho_m V' + \rho_e) d \Phi + (\rho_m \Omega r^2 - \rho_e \psi) d \Omega + dp = 0$, where the prime stands for the derivative with respect to $\Phi$. Therefore, the pressure of the fluid can be considered as a function of $\Phi$ and $\Omega$, i.e., $p = p(\Phi, \Omega)$. This allows us to state two new theorems.
\begin{teo}[\textit{new}]\label{teo:newt}
If a differentially rotating charged pressure fluid is of Weyl-type and is in equilibrium in the Newton-Maxwell theory, and obey the constraint $\rho_m V' + \rho_e \neq 0$, then if any   
pair among the three surfaces of constant $\Phi$, $\Omega$, and $p$ coincide, the third also coincides with the other two.
\end{teo}

\begin{proof}
By assumption, the equilibrium equation reads $(\rho_m V' + \rho_e) d \Phi + (\rho_m \Omega r^2 - \rho_e \psi) d \Omega + dp = 0$, then one can see that if any two of the three differentials $d\Phi$, $d \Omega$, and $dp$ are zero, then the third differential is also zero, meaning that the three surfaces of constant $\Phi$, $\Omega$, and $p$  coincide.
\end{proof}

This theorem generalizes Theorem \ref{teo:2a} by including the fluid pressure.
It is worth noticing that, unlike the term $\rho_m \Omega r^2 - \rho_e \psi$ being non-negative and vanishing only
in the static case, there is no restriction on the term $\rho_m V' + \rho_e$ such
that nothing prevents the case where $\rho_m V' + \rho_e = 0$. In this case, we have another theorem.

\begin{teo}[\textit{new}]\label{teo:newt1}
If a differentially rotating charged pressure fluid is of Weyl-type, is in equilibrium in the Newton-Maxwell theory, and obeys the relation $\rho_m V' + \rho_e = 0$, then the surfaces of constant angular velocity coincide with the surface of constant pressure, and vice versa.
\end{teo}
\begin{proof}
 Taking into account the hypothesis $\rho_m V' + \rho_e = 0$, the equilibrium condition \eqref{eq:equilp} reduces to $(\rho_m \Omega r^2 - \rho_e \psi)d \Omega + dp = 0$. Therefore, it is straightforward to see that if any one of the differentials $d\Omega$ or $d p$ is zero, then both of them are zero, meaning that the surfaces of constant $\Omega$ and $p$ coincide.  
\end{proof}

This theorem is similar to Theorem \ref{teo:p0rigidN1} since it implies that the potentials $V$ and $\Phi$ are functionally related, i.e., $V=V(\Phi)$.

\subsubsection{Rigidly rotating charged pressure fluids in the Newton-Maxwell theory}

\label{subsection:newtweylfluid}
Here we consider rotating charged fluid with nonzero pressure in rigid rotation, for which $d\Omega=0$. 
In this case, written in terms of total derivatives, the equilibrium equation~(\ref{eq:equi1}) gives us  $\rho_m d V + \rho_e d \Phi + d p = 0$. It is then seen that $V$, $\Phi$ and $p$ are functionally related. For instance, we may assume that $p = p(V, \Phi)$, with $\rho_m = - (\partial p/ \partial V)_{\Phi}$ and $\rho_e = - (\partial  p/ \partial \Phi)_{V}$. Therewith, we can state the following theorem, whose static (nonrotating) general relativistic version can be found in \cite{Guilfoyle:1999yb}.

\begin{teo}[\textit{rigidly rotating and axisymmetric Newtonian version of Guilfoyle 1999}] \label{teo:4}

For any rigidly rotating charged pressure fluid in equilibrium in the Newton-Maxwell theory, if any pair of the three surfaces of constant $V$, $\Phi$, and $p$ coincide, then the third also coincides with the other two.
\end{teo}

\begin{proof}
The proof is similar to Theorem \ref{teo:4dif} and it follows by noticing that the condition of constant $V$, $\Phi$,  or $p$ defines a level surface. Moreover, since $\rho_m d V + \rho_e d \Phi + d p = 0$, it is straightforward to see that if any pair of the differentials $dV,\, d\Phi$, and $d p$ vanish, then all of them vanish.
\end{proof}

Theorem \ref{teo:4} generalizes Theorem \ref{teo:p0rigidN1} by including the fluid pressure, and it states the analogous result as for the nonrotating charged pressure fluids studied in Ref.~\cite{Lemos:2009}. This coincident result can be understood by noting that a constant angular velocity $\Omega$ does not contribute to the equilibrium equation.

\subsubsection{Rigidly rotating charged pressure fluids of Weyl-type in the Newton-Maxwell theory}
 
Now, let us consider the Weyl ansatz $V = V(\Phi)$ and see its consequences for a rigidly rotating fluid with nonzero pressure.  In this case, the equilibrium equation \eqref{eq:equi1} can be rewritten as $\left(\rho_m V' + \rho_e\right) d \Phi + d p = 0$.  If $\rho_m V' + \rho_e \neq 0$, the pressure results also a function of $\Phi$, i.e., $p = p(\Phi)$. This allow us to state a new theorem.

\begin{teo}[\textit{rigidly rotating and axisimmetric Newtonian version of Guilfoyle 1999}] \label{teo:5}
(i) If a rigidly rotating charged pressure fluid is of Weyl-type and is in equilibrium in the Newton-Maxwell theory, then the equipotentials surfaces are also surfaces of constant pressure, and vice versa.

(ii) If a rigidly rotating charged fluid with pressure is of Weyl-type and is in equilibrium in the Newton-Maxwell theory, then either the pressure gradient vanishes at the surface of the fluid, or the surface is an equipotential.
\end{teo}
\begin{proof}
To prove assertion \textit{(i)} we use the Weyl ansatz $V = V(\Phi)$ and the equilibrium condition in the form $(\rho_m V' + \rho_e) d \Phi + d p = 0$. Then, the surfaces of constant $V$ and $\Phi$ coincide since $d V = V' d \Phi$, and Theorem \ref{teo:4} implies that the surface of constant $p$ also coincides. On the other hand, with the condition $\rho_m V' + \rho_e \neq 0$ obeyed, a surface of constant $p$ implies in $d \Phi = 0$ and thus, since $d V = V' d \Phi$, it also implies that $d V = 0$. To prove assertion \textit{(ii)}, we first note that the boundary conditions require that the pressure is zero at the surface of the fluid distribution. Therefore, the surface of the fluid is a surface of constant pressure, i.e., $d p = \mathbf{\nabla}_i p\, d x^i = 0$ on such a surface. This implies that either the pressure gradient vanishes at the surface, or by assertion \textit{(i)} the surface of the fluid is an equipotential surface.
\end{proof}

This theorem is the rigidly rotating Newtonian version of a theorem for relativistic nonrotating charged fluids found in Ref. \cite{Guilfoyle:1999yb}.

Other properties of rigidly rotating charged pressure fluids of Weyl-type in the Newton-Maxwell theory may be drawn by a deeper analysis of the field equations.  By assuming the Weyl ansatz $V=V(\Phi)$, from Eqs.~(\ref{eq:pot4}) and (\ref{eq:pot5}) and using the equilibrium condition, it follows an equation for $\Phi$ in terms of derivatives of $V$ and $p$, namely,
\begin{align}
    \big(V'^2 -1\big) \nabla^2 \Phi & + V'V'' \big(\nabla_i \Phi\big)^2 + 4 \pi p' \nonumber \\ & + \Omega\big( \nabla^2 \psi + 2 \Omega V'\big) = 0. \label{eq:field4}
\end{align}
Hence, given $V(\Phi)$ and $p(\Phi)$ this equation can be solved for $\Phi$ and the other quantities may be immediately calculated. For $\Omega = 0$, Eq.~(\ref{eq:field4}) reduces to the same equation obtained by Lemos and Zanchin for the static case (see Eq.~(26) in \cite{Lemos:2009}). 

Now, turning once again to Eqs.~(\ref{eq:pot4}) and (\ref{eq:pot5}) and following \cite{Lemos:2009}, we obtain 
\begin{align}
    \nabla^2\big(V + \epsilon \beta \Phi - \gamma\big) = &\, 4 \pi\Big(\rho_m - \frac{\Omega^2}{2\pi} \nonumber \\\, &- \epsilon \beta \rho_e\big(1 -  \Omega^2 r^2\big) + \epsilon \beta\rho_{GJ}\Big), \label{eq:field5}
\end{align}
where $\beta$ and $ \gamma$ are arbitrary constants, and $\epsilon = \pm 1$.
With this equation in hand, we can state a new theorem that is the rigidly rotating version of a result found by Lemos and Zanchin \cite{Lemos:2009}.

\begin{teo}[\textit{rigidly rotating and axisymmetric version of Lemos and Zanchin 2009}] \label{teo:7}
 (i) If a rigidly rotating charged pressure fluid is in equilibrium in the Newton-Maxwell theory, and obeys the relation  $V + \epsilon \beta \Phi - \gamma = 0$, where $\beta$ and $\gamma$ are constants, then it follows the constraint
 \begin{align}
     \rho_m - \frac{\Omega^2}{2\pi} - \epsilon \beta \rho_e(1 -  \Omega^2 r^2) + \epsilon \beta\rho_{GJ}= 0. \label{eq:equicon}
 \end{align}
 (ii) If a rigidly rotating charged pressure fluid is in equilibrium in the Newton-Maxwell theory, obeys the relation $\rho_m - \frac{\Omega^2}{2\pi} - \epsilon \beta \rho_e\left(1 -  \Omega^2 r^2\right) + \epsilon \beta\rho_{GJ}= 0$, and there is a closed surface with no singularities, holes, or alien matter inside, and where $V+ \epsilon \beta \Phi - \gamma $ vanishes, then it follows that
 \begin{align}
     V = - \epsilon \beta \Phi + \gamma \label{eq:linearweyl}
 \end{align}
 everywhere.
\end{teo}
\begin{proof}
The proof of assertion \textit{(i)} is straightforward. By hypothesis, if $V + \epsilon \beta \Phi - \gamma = 0$ everywhere in the fluid, it follows that the right-hand side of Eq.~(\ref{eq:field5}) must be zero, given Eq.~\eqref{eq:equicon} as a result. The proof of assertion \textit{(ii)} is similar to the proof of Theorem \ref{teo:3}\textit{(ii)}.  From Eq.~(\ref{eq:field5}), the assumption $\rho_m - \frac{\Omega^2}{2\pi} - \epsilon \beta \rho_e\left(1 -  \Omega^2 r^2\right) + \epsilon \beta\rho_{GJ}= 0$ implies in $\nabla^2\big(V + \epsilon \beta \Phi - \gamma\big) = 0$, i.e., $F \equiv V + \epsilon \beta \Phi - \gamma$ is a harmonic function. Integrating the divergence $\nabla_i \left(F \nabla^i F\right)$ over a finite volume in the Euclidean space $\mathbb{R}^3$, it gives
\begin{align}
    \int_{\mathbb{V}_{\!S}}\!\! \nabla_i \left(F \nabla^i F\right) d \mathbb{V} = \int_{\mathbb{V}_{\!S}} \left(\nabla_i F\right)^2 d\mathbb{V} = \int_{\!S} \big(F \nabla_i F\big) n^i dS, 
\end{align}
with $S$ being the boundary of the finite volume $\mathbb{V}_{\!S}$, $n^i$ being the unit vector normal to $S$, and the Gauss theorem has been used. Now, if there is a closed surface where $F = 0$, then by identifying such a surface with $S$ one finds $\int_{\mathbb{V}_{\!S}} \big(\nabla_i F\big)^2 d\mathbb{V}  = 0$, which is satisfied only if $\nabla_i F=0$ in the whole volume $\mathbb{V}_{\!S}$. Therefore, $F = \text{constant}$ throughout the volume $\mathbb{V}_S$ of the fluid.
\end{proof}

Analogously to what has been noted by Lemos and Zanchin for the nonrotating case~\cite{Lemos:2009}, Eq.~(\ref{eq:equicon}) can be seen as the most general condition relating the densities of a rigidly rotating charged fluid with pressure in the Newton-Maxwell theory obeying the Weyl-type relation given by Eq.~(\ref{eq:linearweyl}). In the nonrotating case we have $\Omega = 0$ and it follows $ 
    \nabla^2\big(V + \epsilon \beta \phi - \gamma\big) = 4 \pi \left(\rho_m - \epsilon \beta \rho_e\right)$,
thus recovering the result obtained by Lemos and Zanchin \cite{Lemos:2009}.

Lastly, we notice that, in the rigidly rotating case, by assuming assertion \textit{(i)} of the Theorem~\ref{teo:7} and from Eq.~(\ref{eq:field4}), the following relation between $p'(\Phi), \rho_m$ and $\rho_e$ must be satisfied
\begin{align}
    p' + \rho_e = \epsilon \beta \rho_m.
\end{align}
A similar relation was obtained by Lemos and Zanchin \cite{Lemos:2009} for the static case. Besides, for rigidly rotating charged fluids of Weyl-type with nonzero pressure, $\rho_m$ and $\rho_e$ are directly proportional only when $p(\Phi)$ is a constant. In such a case, one has $\rho_e = \epsilon \beta \rho_m$, which recovers the results discussed at the end of Sec.~\ref{sec:2b} for a rigidly rotating fluid with zero pressure.

\section{Rotating Einstein-Maxwell charged fluids with pressure}
\label{sec:3}

\subsection{The model and the basic equations}
\label{sec:3a}

With the aim of exploring some general properties of rotating charged fluids in the Einstein-Maxwell theory, in this section we write the basic equations governing such a kind of systems.  For the sake of comparison, we follow closely the works by Islam \cite{Islam1978} and Bonnor~\cite{Bonnor1980b}, with the key difference that we also include pressure in the rotating charged fluid. The relevant equations are obtained from the Einstein-Maxwell system of equations, 
\begin{align}
    & G_{\mu \nu} =  8\pi\left(E_{\mu \nu} +M_{\mu\nu}\right),\label{eq:Einst}\\
    &\nabla_{\nu} F^{\mu \nu} = 4 \pi J^{\mu}, \label{eq:Maxw}
\end{align}
where Greek indices range from $0$ to $3$. $G_{\mu \nu}= R_{\mu \nu}- \frac{1}{2}g_{\mu \nu} R$ is the Einstein tensor,  with $R_{\mu\nu}$ being the Ricci tensor, $g_{\mu \nu}$ being the metric tensor, and $R$ being the Ricci scalar. $T_{\mu \nu}$ stands for the energy-momentum tensor. $F^{\mu\nu}$ is the Faraday-Maxwell strength tensor which may be written in terms of a gauge vector potential $A_\mu$ as $F_{\mu \nu} =  \nabla_{\nu} A_{\mu} - \nabla_{\mu} A_{\nu}$, while $\nabla_{\mu}$ is the covariant derivative compatible with the four-dimensional Lorentzian metric, and $J^\mu$ is the current density. 

The energy-momentum tensor is composed by two parts, $T_{\mu \nu} = E_{\mu \nu} + M_{\mu \nu}$, the first part coming from the electromagnetic field and the second one from the matter itself. The electromagnetic part $E_{\mu \nu} $ is given by
\begin{align} 
    E_{\mu \nu} = \frac{1}{4\pi}\left( F_{\mu \alpha} F_{\nu}^{\ \alpha} - \frac{1}{4}g_{\mu \nu} F_{\alpha \beta}F^{\alpha \beta}\right).
\end{align}
Meanwhile, $M_{\mu \nu}$ is the fluid energy-momentum tensor given by
\begin{align}
    M_{\mu \nu} = \big(\rho_m + p\big)u_{\mu}u_{\nu} + p g_{\mu \nu},
\end{align}
where $\rho_m$ is the matter energy density,  $p$ is the fluid pressure,  and $u^{\mu}$ is the fluid four-velocity satisfying $u^{\mu}u_{\mu} = -1$. 

Additionally, we consider just a convective  current density of the form 
\begin{equation} 
J^{\mu} = \rho_e u^{\mu}, \label{eq:current}
\end{equation}
with $\rho_e$ standing for the electric charge density.

We consider stationary and axisymmetric spacetimes such that, given a coordinate system of cylindrical type $(x^0, x^1, x^2, x^3) = (t, r, z, \varphi)$, the metric can be written in the form \eqref{eq:metric3a}, which we rewrite here for convenience, 
\begin{align}
    ds^2 = -f\, dt^2 + 2k\, dt d\varphi + l\, d \varphi^2 + e^{\mu}\left(dr^2 + dz^2\right) \label{eq:metric3}
\end{align}
with the metric coefficients being functions of $r$ and $z$ only, i.e., $f= f(r,z)$, $k= k(r,z)$, etc. This form of the metric instead of \eqref{eq:metric2} is convenient for the present analysis, in particular for comparison of our results with previous works in the literature.  In terms of the metric functions appearing in \eqref{eq:metric2}, we have the relations  $W^2=f$,
$\omega_i = -f^{-1} k \delta_i^\varphi= \left(0,\,0,\,\omega_\varphi\right)$, $h_{ij}= f\,{\rm diag}\big(e^\mu f, \,e^{\mu}f,\,  k^2+ lf \big)$, where $\delta_i^j$ is the Kronecker delta tensor. 

With the metric in the form \eqref{eq:metric3}, the gauge potential and the velocity of the rotating charged fluid may be written, respectively, as
\begin{align}
A_{\mu} =&\, \phi \delta_{\mu}^{\ t} + \psi \delta_{\mu}^{\ \varphi}, \label{eq:gaugepot}\\
    u^{\mu} =&\, {\cal F}^{-1}\left(\delta^{\mu}_{\ t} + \Omega \delta^{\mu}_{\ \varphi}\right),\label{eq:4velocity}
\end{align}
where 
\begin{equation} 
\mathcal{F}^2  \equiv  f - 2k\Omega - l \Omega^2, \label{eq:calF}
\end{equation}
with $\Omega \equiv u^{\varphi}/ u^t = d \varphi/ dt$ being the angular velocity of the fluid, and $\phi=\phi(r,z)$ and $\psi=\psi(r,z)$ are two functions representing the electric and the magnetic potentials, respectively. The function $\mathcal{F}$ is also called the redshift factor \cite{thorne}.

The important geometric quantities for the present analysis are the three metric potentials $g_{tt} = -f$, $g_{t\varphi} = k $, and $g_{\varphi \varphi} = l$.
From the Einstein field equations \eqref{eq:Einst} we get three nontrivial independent equations for these three metric potentials,
\begin{align}
&  \partial_j\left(D^{-1} \partial_j f\right) + D^{-3}f \left(\partial_j f \partial_j l + \left(\partial_j k\right)^2\right) \nonumber \\
    &\hskip 1cm  = 16 \pi e^{\mu} D^{-1}\left(M_{tt} + E_{tt} + \frac{1}{2} f M_\mu^\mu\right), \label{eq:Einst1} \\
    &  \partial_j\left(D^{-1} \partial_j k\right) + D^{-3}k \left(\partial_j f \partial_j l + \left(\partial_j k\right)^2\right)  \nonumber \\ & \hskip 1cm =  -16 \pi e^{\mu} D^{-1}\left(M_{t\varphi} + E_{t\varphi} - \frac{1}{2} k M_\mu^\mu\right), \label{eq:Einst2} \\ 
    & \partial_j(D^{-1} \partial_j l) + D^{-3}l \left(\partial_j f \partial_j l + \left(\partial_j k\right)^2\right)\nonumber \\
    &\hskip 1cm = -16 \pi e^{\mu} D^{-1}\left(M_{\varphi \varphi} + E_{\varphi \varphi} - \frac{1}{2} l M_\mu^\mu\right), \label{eq:Einst3}
\end{align}
where we defined the quantity $D$ as
\begin{equation}
D^2 \equiv fl + k^2, \label{eq:D2}
\end{equation}
with the Roman indexes $i, \, j,\, $ etc., ranging from $1$ to $2$, i.e., $x^1 = r$, $x^2 = z$.
Besides, the repeated covariant  indices in the $r-z$ plane are to be summed over as well. 

The nontrivial components of the tensors $M_{\mu \nu}$ and $E_{\mu \nu}$, as well as the trace of the matter energy-momentum tensor $M_\mu^\mu$,  appearing on the right-hand side of Eqs.~\eqref{eq:Einst1}--\eqref{eq:Einst3} are given in Appendix \ref{sec:app1}, see Eqs.~\eqref{eq:Mtt}--\eqref{eq:Ett}.
Some other important relations regarding the electromagnetic fields and the related energy density are also presented in Appendix \ref{sec:app1}. 

The important electromagnetic fields for the present analysis are the two scalar potentials, the electric potential $\phi$ and the magnetic potential $\psi$. 
In turn, from the Maxwell field equations \eqref{eq:Maxw} we get two independent equations for $\phi$ and $\psi$, namely,
\begin{align}
 & k \nabla^{\dagger 2}_{-} \psi - l \nabla^{\dagger 2}_{-} \phi + \partial_j k \partial_j \psi - \partial_j l \partial_j \phi = 4\pi r^2 e^{\mu} \rho_e \mathcal{F}^{-1} ,  \label{eq:Maxa}\\
& f \nabla^{\dagger 2}_{-} \psi + k \nabla^{\dagger 2}_{-} \phi + \partial_j k \partial_j \phi + \partial_j f \partial_j \psi = 4\pi r^2 e^{\mu} \rho_e \mathcal{F}^{-1} \Omega, \label{eq:Maxb} 
\end{align}
where the operators $\nabla^{\dagger 2}$ and $ \nabla^{\dagger 2}_-$ are defined, respectively, by 
\begin{align}
    &\nabla^{\dagger 2} = \partial_r^2 + \frac{\partial_r D}{D} \partial_r + \partial_z^2 + \frac{\partial_z D}{D} \partial_z, \label{eq:nabla1} \\
    & \nabla^{\dagger 2}_{-} = \partial_r^2 - \frac{\partial_r D}{D} \partial_r + \partial_z^2 - \frac{\partial_z D}{D} \partial_z. \label{eq:nabla*}
    \end{align}
 
The last important equations for the present analysis come from the conservation of the energy-momentum tensor, i.e., $\nabla_{\nu} T^{\mu \nu} = 0$, which yields two equations
\begin{align}
    & \big(\rho_m + p\big) \partial_j \mathcal{F} + \rho_e \partial_j \Phi \nonumber  \\
    & +  \left[\big(\rho_m + p \big)\mathcal{F}^{-1}K -  \rho_e \psi \right]\partial_j \Omega + \mathcal{F} \partial_j p = 0,
    \label{eq:relequi}
\end{align}
where the index $j$ ranges from $1 $ to $2$, $\mathcal{F}$ is the redshift function defined in Eq.~\eqref{eq:calF}, and where we defined the effective electromagnetic potential $\Phi$ and the relativistic centrifugal potential $K$ by
\begin{align}
    & \Phi = \phi+ \Omega\, \psi, \label{eq:Phi}\\
    & K = k +\Omega\, l, \label{eq:K}
\end{align}
respectively. 
Equation~(\ref{eq:relequi}) generalizes the equilibrium equation obtained in \cite{Islam1978, Bonnor1980b} here by taking into account the fluid pressure.

By comparing Eqs.(\ref{eq:equi1}) and (\ref{eq:relequi}), one can see that $\mathcal{F}$ plays the same role of the effective Newtonian potential $V$, given in Eq.~\eqref{eq:eff-gravpot}. In fact, in \cite{Bonnor1980b} Bonnor has shown that, in the weak-field limit, $\mathcal{F}$ reduces to $V$. For the sake of comparison, we write $f=1+2U$ and expand ${\cal F}$ up to the second order in the quantities $U$, $k$, and $\Omega$, that are considered small when compared to unity. The result is
\begin{equation}
    {\cal F}= 1 + U -  \Omega k - \frac12 \Omega^2 l\equiv 1+ V_{rel},  \label{eq:Fapprox}
\end{equation}
where terms like $\Omega^2 k^2$ have been neglected. Therefore, when compared to the corresponding potential in the Newton-Maxwell theory given by $ V= U-\dfrac12 \Omega^2 r^2$, the relativistic gravitational potential \begin{equation} \label{eq:Vrel}
V_{rel}= \mathcal{F}-1= U - \frac12\Omega \left( 2k +\Omega l\right)
\end{equation} 
has the contribution of the shift function $k$. This difference is a consequence of the fact that the shift function has no analogue counterpart in the Newtonian theory. 

Notice also that the effective relativistic electromagnetic potential $\Phi=\phi+\Omega\psi$ in the Einstein-Maxwell theory, given in Eq.~\eqref{eq:Phi}, has the same form of the corresponding effective electromagnetic potential in the Newton-Maxwell theory, given in Eq.~\eqref{eq:eff-electpot}.  

It is worth also mentioning that the quantity $\big(\rho_m+p\big){\cal F}^{-1}K = \big(\rho_m+p\big)u_{\varphi} $ in Eq.~\eqref{eq:relequi} can be interpreted as the angular momentum density of the fluid \cite{abra}. Analogously to the rotating charged fluid in the Newton-Maxwell theory, the quantity $\rho_e\psi$ is interpreted as the angular momentum density of the electromagnetic field. 

Now the relevant equations for the relativistic problem are  the Einstein equations~(\ref{eq:Einst1}), (\ref{eq:Einst2}), and (\ref{eq:Einst3}), the Maxwell equations \eqref{eq:Maxa} and \eqref{eq:Maxb},  and the equilibrium equation for the system given by Eq.~(\ref{eq:relequi}).

\subsection{Zero-pressure rotating charged fluids in the Einstein-Maxwell theory}
\label{sec:3b}

\subsubsection{Differentially rotating charged dust fluids in the Einstein-Maxwell theory}

Let us begin the analysis by considering rotating charged dust fluids with differential rotation, for which $p = 0$, in the Einstein-Maxwell theory. Some general properties of this kind of systems have already been investigated in previous works, see e.g. \cite{Raycha,RayDe}. Here we investigate further such systems and present new results.

In the case of zero pressure, and after contracting with $d x^j$, the equilibrium Eq.~(\ref{eq:relequi}) reduces to
\begin{align}
     \rho_m d \mathcal{F} + \rho_e d \Phi + \Big(\rho_m \mathcal{F}^{-1}K -  \rho_e \psi \Big)d \Omega = 0. \label{eq:relequi2}
\end{align}
Equation \eqref{eq:relequi2} implies that, for  $\rho_m \mathcal{F}^{-1}K -  \rho_e \psi \neq 0$, the metric potential $\mathcal{F}$, the electromagnetic potential $\Phi$, and the angular velocity $\Omega$ are functionally related. Similarly to the rotating charged fluid in the Newton-Maxwell theory, we shall assume that the term $\rho_m \mathcal{F}^{-1}K -  \rho_e \psi$ vanishes only in the static limit. With such an assumption, we can state a new theorem that is the general relativistic version of Theorem \ref{teo:2}.

\begin{teo}[\textit{new}] \label{teo:rel2}
For any distribution of a differentially rotating charged dust in the Einstein-Maxwell theory obeying the constraint $\rho_m \mathcal{F}^{-1}K -  \rho_e \psi \neq 0$, if any two out of the three surfaces of constant $\mathcal{F}$, $\Phi$, and $\Omega$ coincide, then the third surface also coincides.
\end{teo}

\begin{proof}
The proof is similar to the corresponding theorem stated for the Newton-Maxwell theory. Since $\mathcal{F}$, $\Phi$, and $\Omega$ are scalar fields in a four-dimensional spacetime and since they do not depend upon time, we can consider spacelike hypersurfaces $\Sigma_t$ defined by $t = \text{constant}$ such that the conditions of constant $\mathcal{F}$, $\Phi$, and $\Omega$ define two-dimensional spatial surfaces on each $\Sigma_t$. Now, from Eq.~(\ref{eq:relequi2}) if any two of the differentials $d \mathcal{F}$, $d\Phi$, and $d\Omega$ vanish, then all of them vanish, implying that the three surfaces coincide.
\end{proof}

The problem of finding exact solutions for differentially rotating charged dust fluids in the Einstein-Maxwell theory was investigated by Islam \cite{Islam1978, Islam1979, Islam1983}, as well  by Van der Bergh, Wils and Islam \cite{Bergh, Wils, Islam1984}. In the case of vanishing Lorentz force, we have $\partial_j \Phi - \psi \partial_j \Omega = 0$, which implies that $\Phi$ is a function of $\Omega$, $\Phi =  \Phi(\Omega)$ with $ d \Phi/ d \Omega=\psi$. Therewith, Eq.~(\ref{eq:relequi2}) implies that $\mathcal{F}$ and $\Omega$ are functionally related with $K = - \frac{1}{2} d F / d \Omega$, where $F = \mathcal{F}^2$. For more details on this particular case see \cite{Wils, Islam1984}. Here we are not interested in finding explicit solutions for these kind of systems and so we move on to study other general properties they present. However, it is worth pointing out that the solutions obtained in \cite{Wils, Islam1984} satisfy Theorem~\ref{teo:rel2}.

\subsubsection{Differentially rotating charged dust fluids of Weyl-type in the Einstein-Maxwell theory}

Here we consider the Weyl ansatz between the redshift function $\mathcal{F}$ and the effective electromagnetic potential $\Phi$, i.e., a assume these two potentials are functionally related, $ \mathcal{F} = \mathcal{F}(\Phi)$, and explore some of its consequences for a differentially rotating charged dust.  The equilibrium equation now reads 
\begin{align}
    \big(\rho_m \mathcal{F}' + \rho_e\big) d \Phi + \Big(\rho_m \mathcal{F}^{-1}K -  \rho_e \psi\Big) d \Omega = 0 ,\label{eq:relequi3}
\end{align}
where the prime stands for the derivative with respect to $\Phi$. Moreover, we assume that $\rho_m \mathcal{F}^{-1}K -  \rho_e \psi  \neq 0$ and also $\rho_m \mathcal{F}' + \rho_e \neq 0$. Therefore, the angular velocity of the fluid is also a function of $\Phi$, i.e., $\Omega = \Omega(\Phi)$. This allow us to state the following theorem, which is the general relativistic version of Theorem \ref{teo:2a}.

\begin{teo}[\textit{new}] \label{teo:rel2a} If a differentially rotating charged dust fluid is of Weyl-type and is in equilibrium in the Einstein-Maxwell theory, then the equipotentials surfaces are also surfaces of constant angular velocity, and vice versa.
\end{teo}

\begin{proof}
By using the Weyl ansatz $\mathcal{F} = \mathcal{F}(\Phi)$ and Eq.~(\ref{eq:relequi3}) we verify that the surfaces of constant $\mathcal{F}$ and $\Phi$ coincide, and then Theorem \ref{teo:rel2} implies that the surface of constant $\Omega$ also coincides. On the other hand, a surface of constant $\Omega$ implies $d \Phi = 0$ and thus, by the Weyl ansatz, it also implies in $d \mathcal{F} = 0$ completing the proof. 
\end{proof}

At last we consider the particular case where $\rho_m \mathcal{F}' + \rho_e = 0$.  
Therefore, since  for $\rho_m \mathcal{F}^{-1}K -  \rho_e \psi  \neq 0$, Eq.~\eqref{eq:relequi3} implies in $d \Omega = 0$ and, in this case, the rotating charged dust fluids are necessarily in rigidly rotation.

\subsubsection{Rigidly rotating charged dust fluids in the Einstein-Maxwell theory}

Given that the fluid is in rigid rotation $d\Omega=0$, besides being a dust fluid, Eq.~(\ref{eq:relequi2}) reduces to
\begin{align}
    \rho_m d \mathcal{F} + \rho_e d \Phi  =0\label{eq:relequirig}.
\end{align}
From this relation it is possible to state the general relativistic version of our Newtonian Theorem \ref{teo:p0rigidN1}.

\begin{teo}[\textit{rigidly rotating and axisymmetric version of De and Raychaudhuri 1968}]\label{teo:rel1}
For any rigidly rotating charged dust distribution in the Einstein-Maxwell theory, the surfaces of constant redshift $\mathcal{F}$ coincide with the surfaces of constant $\Phi$, and $\mathcal{F}$ is functionally related to $\Phi$, i.e., $\mathcal{F} = \mathcal{F}(\Phi)$.
\end{teo}
\begin{proof}
The condition of constant $\mathcal{F}$ or constant $\Phi$ defines a two-dimensional surface in the spacetime. 
Equation~(\ref{eq:relequirig}) implies that if any one of the differentials $d \mathcal{F}$ or $d\Phi$ is zero, then both of them vanish and the surfaces of constant $\mathcal{F}$  and $\Phi$ coincide.
Besides, the equilibrium equation  \eqref{eq:relequirig} gives $d \mathcal{F}/d\Phi = - \rho_e/ \rho_m$ what implies that both scalar fields are functionally related, namely, $\mathcal{F} = \mathcal{F}(\Phi)$.
\end{proof}

Theorem \ref{teo:rel1} is the rigidly rotating version of a theorem due to De and Raychaudhuri \cite{deray68}, that has been considered also in Ref.~\cite{Lemos:2009}. Similarly to the nonrotating case, this theorem implies that, for rigidly rotating charged dust fluid distributions, the relation $\mathcal{F} = \mathcal{F}(\Phi)$ is a consequence of the equilibrium condition. 
Hence, the relation $\mathcal{F} = \mathcal{F}(\Phi)$ may be interpreted as a generalized version of the Weyl ansatz for general relativistic rigidly rotating systems. With this nomenclature, Theorem \ref{teo:rel1} may be rephrased as {\it rigidly rotating charged dust fluid are necessarily Weyl-type systems}. 

The ansatz $\mathcal{F} = \mathcal{F}(\Phi)$ was made a priori by Bonnor \cite{Bonnor1980b}. 
More specifically, in order to obtain a solution of the system of equations for a rigidly rotating charged fluid in the Einstein-Maxwell theory, Bonnor considered the linear relation  $\mathcal{F} = - \epsilon \sqrt{\alpha}\, \Phi$. Furthermore, Bonnor assumed also that the ratio between the metric potentials $K$ and $\mathcal{F}$ would be a function of the magnetic potential $\psi$ alone, i.e., $K/\mathcal{F} = g(\psi)$ and, to further simplifying the system of equations, he made the simplest choice $g(\psi) = \epsilon \sqrt{\alpha}\, \psi$, and then took $\alpha = 4$. Later on, Raychaudhuri \cite{Raycha} realized that this linear relation assumed by Bonnor is equivalent to choosing a linear relation between the fluid velocity $u_\mu$ and the electromagnetic gauge potential $A_\mu$, viz,  
\begin{align}
    u_{\mu} = -a\, A_{\mu}. \label{eq:bonnor}
\end{align}
with $a$ being a constant parameter. In our notation, $a = - \epsilon \sqrt{\alpha}$. From now on we refer to this choice involving the two relations $\mathcal{F} = - \epsilon \sqrt{\alpha}\, \Phi$ and  $K= \epsilon \sqrt{\alpha}\, \psi\mathcal{F}$ as the Bonnor-Raychaudhuri ansatz. Additionally, in the case of a rigidly rotating fluid, Raychaudhuri \cite{Raycha} showed that the expansion and the shear of the fluid vanish, thus satisfying the shear-free conjecture \cite{Collins}. Besides, as noticed by Som and Raychaudhuri \cite{Som}, by Islam \cite{Islam1977}, and by Bonnor \cite{Bonnor1980b}, in the case of vanishing Lorentz force, the equilibrium condition \eqref{eq:relequirig} implies that $\mathcal{F}$ is a constant.

\subsubsection{Rigidly rotating charged dust fluids of Weyl-type in the Einstein-Maxwell theory}
\label{sec:relrigid}

Further interesting properties of rigidly rotating charged dust fluids in the Einstein-Maxwell theory may be drawn through a deeper analysis of the full set of field equations \eqref{eq:Einst1}--\eqref{eq:Maxb}, after being adapted to this particular kind of fluid. 
For such a task we follow Islam \cite{Islam1978} and adopt the functions $F\equiv\mathcal{F}^2$, $K$, $L$, $\Phi$, and $\psi$, where $K$ is the effective centrifugal potential \eqref{eq:K}, $L \equiv l$, and $\Phi$ is the effective electromagnetic potential \eqref{eq:Phi}. In the case of a rigidly rotating fluid, this change of variables is equivalent to transforming to a coordinate system whcih rotates with the fluid (see Appendix \ref{sec:app2}). In fact, such a transformation holds also for differential rotation, but the interpretation of the local coordinate transformation as a change of frame is not consistent (see \cite{Islam1978} for more details on this subject). Furthermore, since we are dealing with a zero pressure fluid, we may take $D=r$ in Eq.~\eqref{eq:nabla*} (see also Appendix \ref{sec:app1}).

With the choices just mentioned, the relevant field equations are rewritten in a more convenient form.  In particular, by combining Eqs.~\eqref{eq:Einst1}, \eqref{eq:Einst2}, and \eqref{eq:Einst3} we get an equation for the function $F$,
\begin{align}
      & \nabla^{2} F + \frac{1}{r^2F} \big(F \nabla K -  K \nabla F \big)^2 -  \frac{1}{2 F} \big(\nabla F\big)^2 \nonumber \\ &\qquad\; = 8 \pi F \big(\rho_m + 2\rho_{em}\big), \label{eq:relpot9c} 
\end{align}
where $\nabla^2\equiv \nabla_j \nabla^j$, with $\nabla_j$ standing for the covariant derivative defined in the hypersurface $\Sigma_t$ of constant time $t$ (see Appendix \ref{sec:app2}, cf. Eq.~\eqref{eq:indlaplc}), and the squares indicate contraction of indices on $\Sigma_t$, e.g., $(\nabla F)^2 \equiv \nabla_j F \nabla^j F = e^{-\mu} \left(\partial_j F\right)^2 $. Besides, $\rho_{em}$ is the electromagnetic energy density as measured by a observer comoving with the fluid, whose expression is written explicitly in Appendix \ref{sec:app2}, cf. Eq.~\eqref{eq:rhoem2}. Some details on the derivation of Eq.~\eqref{eq:relpot9c} from Eqs.~\eqref{eq:Einst1}, \eqref{eq:Einst2}, and \eqref{eq:Einst3} are presented in Appendix~\ref{sec:app3}. The electromagnetic energy density $\rho_{em}$ may be written as
\begin{align} \label{eq:rhoem}
        \rho_{em}  = &\frac{e^{-\mu}}{8 \pi r^2F }\Big[\left( r^2 +K^2 \right)\big(\partial_j \Phi\big)^2 + 2 F K \partial_j \Phi \partial_j \psi \nonumber\\ &\qquad\; + F^2 \big(\partial_j \psi\big)^2 \Big].
\end{align}

For the analysis of this section it is useful decomposing $\rho_{em}$ into its electric and magnetic parts, $\rho_{em} = \rho_{el} + \rho_{mg}$, as is done in Appendixes  \ref{sec:app1} and \ref{sec:app2}.  In the case of a rigidly rotating dust fluid, these two part are given in Eqs.~\eqref{eq:rhoel2} and \eqref{eq:rhomg2}, respectively. For convenience, we rewrite the expression for these two parts here,
\begin{align}
    & \rho_{el}\ = \frac{e^{-\mu}}{8 \pi F} \big(\partial_j \Phi\big)^2, \label{eq:rhoel} \\
    & \rho_{mg} = \frac{e^{-\mu}}{8\pi r^2 F}\big(K \partial_j \Phi + F \partial_j \psi\big)^2. \label{eq:rhomg}
\end{align}

In turn, from the combination of the Maxwell equations, Eqs.~\eqref{eq:Maxa} and \eqref{eq:Maxb}, we get
\begin{align}
    & \nabla^{2} \Phi  + \frac{1}{r^2 F}\big(K \nabla F\!- F \nabla K \big) \cdot \left( F \nabla \psi + K \nabla \Phi \right)  \nonumber \\ &\qquad\; - \frac{1}{2 F} \nabla F \cdot  \nabla \Phi = - 4\pi F^{1/2} \rho_e, \label{eq:relpot10c}
\end{align}
where the dot indicates contraction of indexes as, e.g.,  $\nabla F \cdot  \nabla \Phi \equiv \nabla_j F \nabla^j \Phi = e^{-\mu} \partial_j F \partial_j \Phi$, etc. See Appendix~\ref{sec:app3} for more details on the derivation of Eq.~\eqref{eq:relpot10c}.

From Theorem \ref{teo:rel1} it follows that $\mathcal{F} = \mathcal{F}(\Phi)$ and, since $\mathcal{F} = F^{1/2}$, it gives $F=F(\Phi)$. Therefore, by combining Eqs.~(\ref{eq:relpot9c}) and (\ref{eq:relpot10c}) and eliminating $\rho_m$ and $\rho_e$ through the equilibrium equation $\rho_m \mathcal{F}' + \rho_e = 0$, we arrive at
\begin{align}
  &   \big(\mathcal{F}'^2 - 1\big)\nabla^2\Phi + \mathcal{F}'\mathcal{F}''\big(\nabla \Phi\big)^2   \nonumber \\ & \qquad = 8 \pi  {\cal F} \mathcal{F}'\rho_{mg} 
  - \frac{\mathcal{F}'}{2r^2 {\cal F} F} \Big(F \nabla K -  K \nabla F \Big)^2\nonumber \\
  & \quad\qquad + \frac{1 }{r^2 F}\big(K \nabla F\!- F \nabla K \big) \cdot ( F \nabla \psi + K \nabla \Phi ). \label{eq:relpot11}
\end{align}
The left-hand side of Eq.~(\ref{eq:relpot11}) vanishes if the relation between $\mathcal{F}$ and $\Phi$ is of the Majumdar-Papapetrou-type, i.e., if $\mathcal{F} = - \epsilon \Phi + \gamma$ with constant $\gamma$. 

Now, by defining the function $\mathcal{Z} \equiv \sqrt{\mathcal{F}'^2 - 1}$, Eq.~(\ref{eq:relpot11}) can be recast as
\begin{align}
  &   \mathcal{Z} \nabla_j (\mathcal{Z} \nabla^j \Phi) \nonumber \\ &\;\qquad  = 8 \pi  {\cal F} \mathcal{F}'  \rho_{mg}  - \frac{\mathcal{F}'}{2 r^2{\cal F} F} \big(F \nabla K -  K \nabla F \big)^2 \nonumber\\ & \;\quad\qquad + \frac{1 }{r^2 F}\big(K \nabla F- F \nabla K \big) \cdot \big( F \nabla \psi + K \nabla \Phi \big)  . \label{eq:relpot12a}
\end{align}
This equation allows us to state a new theorem.

\begin{teo}[\textit{rigidly rotating and axisymmetric version of Das 1962, De-Raychaudhuri 1968 and Bonnor 1980}] \label{teo:11}

If the right-hand side of Eq.~(\ref{eq:relpot12a}) vanishes, then

(i) In the Einstein-Maxwell theory, if the surfaces of any rigidly rotating charged dust distribution are closed equipotential surfaces and inside these surfaces there are no singularities, holes or alien matter, then  $\mathcal{F}$ and $\Phi$ are related by
\begin{align}
    \mathcal{F} = - \epsilon \Phi + \gamma, \label{eq:lineara}
\end{align}
where $\gamma$ is a constant of integration, $\epsilon = \pm 1$, and it follows that $\rho_e = \epsilon \rho_m$.

(ii) In the Einstein-Maxwell theory, if in a spacetime region the ratio $\rho_e/\rho_m$ equals a constant $\mathcal{K}$, and there are no
singularities, holes or alien matter in that region, then it follows  the relation $\mathcal{F} = - \epsilon \Phi + \gamma$, and $\mathcal{K} = \epsilon.$

\end{teo}
\begin{proof}
To prove assertion \textit{(i)} we first take into account that the right-hand side of Eq.~(\ref{eq:relpot12a}) vanishes to get $\mathcal{Z} \nabla_j (\mathcal{Z} \nabla^j \Phi) = 0$, then define the function $\Psi$ through $\nabla_j \Psi = \mathcal{Z} \nabla_j \Phi$, and then follow the same steps as in Theorem \ref{teo:3}\textit{(i)} for the Newtonian case. The only subtlety here is that, in order to apply Gauss theorem, a finite volume must be properly defined. For this, we consider the spacelike hypersurfaces $\Sigma_t$ of constant $t$ and choose a finite $\mathbb{V}_S$ in each  $\Sigma_t$ (see also Appendix~\ref{sec:app2}). After that, integrating the divergence $\nabla_j \left(\Psi  \nabla^j \Psi \right)$ over the volume $\mathbb{V}_S$ and using Gauss theorem it is shown that $\mathcal{Z} = \sqrt{\mathcal{F}'^2 - 1} = 0$.  This implies in $\mathcal{F} = - \epsilon \Phi + \gamma$, where $\gamma$ is a constant parameter, and the condition $\rho_e = \epsilon \rho_m$ immediately follows. 

The proof of assertion \textit{(ii)} is analogous to the proof of Theorem \ref{teo:3}\textit{(ii)}. Given that $\rho_e / \rho_m = \mathcal{K}$, by the equilibrium condition \eqref{eq:relequirig} it follows the relation $\mathcal{F}' = - \mathcal{K}$. Taking in account that, by assumption,  the right-hand side of Eq.~(\ref{eq:relpot12a}) vanishes, it  gives $ \mathcal{Z} \nabla_j (\mathcal{Z} \nabla^j \Phi) = \mathcal{Z}^2 \nabla^2 \Phi = 0$. Since we are considering a region where there is a charge distribution $\rho_e$, then $\nabla^2 \Phi \neq 0$, remaining that $\mathcal{Z} = \sqrt{\mathcal{F}'^2 - 1} = 0$, and then ${\cal F}=\pm1$. Therefore, it follows that the relation between $\mathcal{F}$ and $\Phi$ is given by Eq.~(\ref{eq:lineara}), while the relation between $\rho_m$ and $\rho_e$ is such that $\rho_e = \epsilon \rho_m$, and then one has $\mathcal{K} = \epsilon$.
\end{proof}

Theorem \ref{teo:11} is the rigidly rotating version of a result by Bonnor \cite{Bonnor:1980nw} for nonrotating axisymmetric charged fluids. It is also the general relativistic version of our Theorem \ref{teo:3}. The first interesting thing to notice regarding this theorem is that, in the nonrotating case for which one has $K =0,\, \psi =0, \, \Omega = 0$, $F = f$, $L = l$, and $\Phi = \phi$, Eq.~(\ref{eq:relpot11}) reduces to
\begin{align}
    \big(\mathcal{F}'^2 - 1\big) \nabla^2 \Phi + \mathcal{F}'\mathcal{F}'' \big(\nabla \Phi\big)^2 = 0.
\end{align}
This equation is equivalent to an equation obtained by Bonnor in \cite{Bonnor:1980nw} and, therefore, Theorem \ref{teo:11} reduces to the result for nonrotating charged dust fluids obtained in that work. Besides, the imposition that the right-hand side of Eq.~(\ref{eq:relpot12a}) vanishes implies a constraint between the magnetic energy density $\rho_{mg}$ and metric and electromagnetic potentials. Such a relation is the relativistic analog of the Newton-Maxwell case displayed on the right-hand side of Eq.~\eqref{eq:rightside}, it involves the relativistic version of the Goldreich-Julian charge density but it is much more intricate and we dot not display the details here to avoid more cumbersome equations.

It is also noteworthy that, in the literature cited in this work, there is no explicitly solution to the Einstein-Maxwell system of equations for rotating fluids in which the relations  $\mathcal{F} = - \epsilon \Phi + \gamma$ and $\rho_e = \epsilon \rho_m$ are obeyed, and the Lorentz force is nonvanishing. This is contrary to the nonrelativistic rotating charged systems where there is a known solution of the Newton-Maxwell system of equations, due to Islam \cite{Islam1983}.
In fact, in the cases with nonvanishing Lorentz force, in order to simplify the set of field equations, it is also considered an additional relation between the functions $K$ and $\psi$,
as in Refs.~\cite{Islam1978, Islam1980, Islam1983, Islam1983a, Bonnor1980b, Raycha, Bergh1984}. In these works, all the solutions with $\rho_e = \epsilon \rho_m$ correspond to static systems. 
Notice, however, that Theorem \ref{teo:11} does not impose any restriction between $K$ and $\psi$. 

On the other hand, for vanishing Lorentz force, there are known solutions where $\mathcal{F} = - \epsilon \Phi + \gamma$ and $\rho_e = \epsilon \rho_m$, with constant $\mathcal{F}$ and $\Phi$, and also with $K$ and $\psi$ being functionally related \cite{Som, Islam1977, Banerjee}. It is somewhat interesting that such solutions do not admit the static limit. In this regard, we shall discuss an example solution due to Islam \cite{Islam1977} in Sec.~\ref{sec:islamsol}.

Finally, in connection with the study of the present section, it is worth mentioning the work by Breipthaupt et al.~\cite{meinel2015}, where a rigidly rotating disc of charged dust fluid that satisfies the Majumdar-Papapetrou relation \eqref{eq:lineara} is studied, and where it is verified that the quasiblack hole limit may be attained.  
The authors also present an interesting relation between the gravitational mass $M$, the angular momentum $J$, the angular velocity $\Omega$, and the baryonic mass of the disc $M_0$.  In our notation, such a relation reads $M - 2 \Omega J = \int_{\Sigma_t}(\mathcal{F} + \epsilon \Phi)d M_0 = \gamma M_0$. This relation is a particular case of the general mass formula obtained by Lemos and Zaslavskii \cite{Lemos:2009wj} for nonextremal rotating charged quasiblack holes.

\subsubsection{Rigidly rotating charged dust fluid of Weyl-type: Islam ansatz}

In order to simplify the set of field equations, inspired by the works of Bonnor \cite{Bonnor1980b}, Islam \cite{Islam1978} and Raychaudhuri \cite{Raycha}, let us consider here the following relations between the metric and electromagnetic potentials, 
\begin{align}
    & \partial_j F = 2 \alpha \Phi \partial_j \Phi, \label{eq:ansatz} \\
    & \partial_j K = -2 \alpha \Phi \partial_j \psi. \label{eq:ansatz1}
\end{align}
These hypotheses may be viewed as a new ansatz, which we shall call the Islam ansatz.

It is easy to see that Eq.~(\ref{eq:ansatz}) implies in $F = \alpha \Phi^2 +\beta$, with $\beta$ being an arbitrary constant. In the particular case with $\beta =0$ and $\alpha =4$, such a relation reduces to the same form assumed by Bonnor \cite{Bonnor1980b}, i.e., ${\cal F} = -2\epsilon\Phi$.  
On the other hand, Eq.~(\ref{eq:ansatz1}) tell us that $K$ is an arbitrary function of $\psi$ only, $K=K(\psi)$, which in general differs from the relation between $K$ and $\psi$ assumed by Bonnor. However, when the Lorentz force vanishes, i.e., for $\partial_j \Phi = 0$, the Islam ansatz given by Eqs.~ (\ref{eq:ansatz}) and (\ref{eq:ansatz1}) is equivalent to the Bonnor-Raychaudhuri ansatz expressed in Eq.~(\ref{eq:bonnor}). 

After making the Islam ansatz, we can state a new theorem, as follows. 

\begin{teo}[\textit{new}]\label{teo:12}
If a rigidly rotating charged dust described by the Einstein-Maxwell theory satisfies the Islam ansatz, given by Eqs.~(\ref{eq:ansatz}) and (\ref{eq:ansatz1}), then the fluid satisfies the following relation
\begin{align}
    \big(\rho_m + 2\left(1 - \alpha\right) \rho_{el} + 2 \rho_{mg}\big)\mathcal{F} + \alpha \Phi \rho_e = 0. \label{eq:stateq}
\end{align}
\end{teo}

\begin{proof}
The proof is straightforward. Given that, by assumption, Eqs.~(\ref{eq:ansatz}) and (\ref{eq:ansatz1}) hold, then it follows the relation $F = \alpha \Phi^2 +\beta$, with constant $\beta$, and this implies in $\nabla^{\dagger 2} F = 2\alpha \Phi \nabla^{\dagger 2} \Phi + 2\alpha \big(\partial_j \Phi\big)^2$. Then, from these results and Eqs.~(\ref{eq:relpot9}) and (\ref{eq:relpot10}), we get
\begin{align}
        2\alpha \Phi \nabla^{\dagger 2} \Phi & + 2\alpha \big(\partial_j \Phi\big)^2 + \frac{1}{r^2} \Big(F \big(\partial_j K\big)^2 - 2 K \partial_j K \partial_j F \nonumber \\ 
        & - L \big(\partial_j F\big)^2 \Big)  = 8 \pi e^{\mu} F \big(\rho_m + 2\rho_{em}\big), \label{eq:relpot9a}\\
         2 \alpha \Phi \nabla^{\dagger 2} \Phi & + \frac{1}{r^2} \Big(F \big(\partial_j K\big)^2 - 2 K \partial_j K \partial_j F - L \big(\partial_j F\big)^2 \Big) \nonumber \\
        & = - 8\pi e^{\mu} \mathcal{F} \alpha \Phi \rho_e. \label{eq:relpot10a}
\end{align}
By replacing Eq.~(\ref{eq:relpot10a}) into Eq.~(\ref{eq:relpot9a}) and introducing $\rho_{el}$ by means of Eq.~(\ref{eq:rhoel}), it follows that $ (\rho_m + 2(1 - \alpha) \rho_{el} + 2 \rho_{mg})\mathcal{F} + \alpha \Phi \rho_e  = 0 $.
\end{proof}

This theorem establishes the energy balance in configurations of rigidly rotating charged dust fluids satisfying the Islam ansatz. The first term in Eq.~(\ref{eq:stateq}), i.e., the term containing $\rho_m$, $\rho_{el} $, $\rho_{mg}$ and ${\cal F}$, 
represents the gravitational energy density,  while the term $\rho_e\Phi$ represents the electromagnetic binding energy of the system \cite{Lemos:2009}.  Some consequences of this theorem are discussed next.

First, consider a system with vanishing Lorentz force.  Then one has $\partial_j \Phi=0$, which means a constant $\Phi$ and, after Eq.~\eqref{eq:rhoel}, a vanishing  $\rho_{el}$,  Eq.~\eqref{eq:ansatz} implies in a constant metric potential $F$, and Eq.~(\ref{eq:stateq}) reduces to $\rho_e = {\rm const.}\times  \left(\rho_m +  2 \rho_{mg}\right) $. This condition is a consequence of the Islam ansatz and, in analogy to the Majundar-Papapetrou condition for static charged dust fluids, we name it the first Islam condition.

Second, for rigidly rotating charged dust fluids with nonvanishing Lorentz force but satisfying the Islam ansatz, we may choose $\mathcal{F} = -\epsilon \sqrt{\alpha}\Phi$, and an alternative form of Eq.~\eqref{eq:stateq} turns out, namely, $\rho_m + 2(1 - \alpha) \rho_{el} + 2 \rho_{mg} -  \epsilon \sqrt{\alpha} \rho_e = 0$. Additionally, the equilibrium condition \eqref{eq:relequirig} reduces to $\rho_e = \epsilon \sqrt{\alpha} \rho_m$ and then Theorem \ref{teo:12} gives  
$ \left(1 - \alpha\right)(\rho_m + 2 \rho_{el}) +  2 \rho_{mg} = 0$, or equivalently $ \rho_m = 2\rho_{mg}/\left(\alpha-1\right) - 2\rho_{el}$. This condition is also a consequence of the Islam ansatz and we name it the second Islam condition. In this case, the energy density $\rho_m$ is related to the electric and magnetic energy densities, and to have the condition $\rho_m>0$ fulfilled one must take $\alpha > 1$ and $\rho_{mg} > (\alpha -1)\rho_{el}$. For $\alpha > 2$ we can see that the magnetic energy density dominates over the electric energy density. This situation was also noticed for systems obeying the Bonnor-Raychaudhuri ansatz for $\alpha = 4$, see e.g. \cite{Raycha}.

A last comment on the consequences of Theorem \ref{teo:12} we make here is regarding the particular case with $\alpha =1$, which does not fit in the analysis presented in the last paragraph.  In such a case, one has a vanishing magnetic field strength $B_\mu=0$ (see Appendix \ref{sec:app2}) and a vanishing magnetic energy density $\rho_{mg} = 0$. Hence, from Eq.~\eqref{eq:rhomg} it follows that $K \partial_j \psi = - F\partial_j \Phi$ which, together with ansatz \eqref{eq:ansatz1} implies in $K^2 = \Phi^4 + \gamma = F^2 + \gamma$, where $\gamma$ is an integration constant. In the case  $\gamma=0$, it results $K=\pm F$ and it can be shown that the charged dust fluid with $\alpha =1$ and satisfying the Islam ansatz does not rotate at all. This result has been already noticed in Refs.~\cite{Islam1977, Raycha, Bergh1984}.

\subsubsection{The Islam's solution: Vanishing Lorentz force}
\label{sec:islamsol}
Here we consider a class of solutions obtained by Islam \cite{Islam1977} for a distribution of rigidly rotating axisymmetric charged dust in which the Lorentz force vanishes, and show that such solutions obey the Islam ansatz. The corresponding class of solutions is characterized by the following relations for the relevant functions,
\begin{align}
    & \Phi = \Phi_0, \ \  F = F_0, \ \ \psi = b \xi, \ \ K = a \xi, \label{eq:phi0} \\
    & L = F^{-1}\left(r^2 - K^2\right) = F_0^{-1}\left(r^2 - a^2 \xi^2\right), \\
    & \partial_r \mu = \frac{1}{2 r}\left(a^2 - 4 b^2 F_0\right)\left(\left(\partial_z \xi\right)^2 - \left(\partial_r \xi\right)^2\right), \\
    & \partial_z \mu = -\frac{1}{r}\left(a^2 - 4 b^2 F_0\right)\partial_r \xi \partial_ z \xi, 
\end{align}
where $\Phi_0, F_0, a$ and $b$ are arbitrary constant parameters, and $\xi$ is a function satisfying $\nabla^{\dagger 2}_- \xi = 0$. For instance, taking $\xi = a_0 r^2 + a_1$, where $a_0$ and $a_1$ are constants, the Islam solution reduces to the Som-Raychaudhuri solution~\cite{Som}. For this solution, one has $K = \frac{a}{b}\psi$ and then the ansatz given by Eq.~(\ref{eq:ansatz1}) is obeyed if the constant parameters satisfy the relation 
\begin{equation}
     \alpha= -\frac{ a} {2b\,\Phi_0}. \label{eq:constraint}
\end{equation}

The energy density and the charge density become
\begin{align}
  & \rho_m = \big(a^2 - 2F_0 b^2\big)\frac{e^{-\mu}}{8\pi r^2}\big(\partial_j \xi\big)^2, \label{eq:rhom} \\
    & \rho_e = ab \mathcal{F}_0 \frac{e^{-\mu}}{4 \pi r^2} \big(\partial_j \xi\big)^2,\label{eq:rhoe}
\end{align}
 respectively.

Additionally, since $\partial_j \Phi = 0$, it follows from Eqs. \eqref{eq:rhoem}, \eqref{eq:rhoel}, and \eqref{eq:rhomg} that $\rho_{el} = 0$ and the electromagnetic energy density reads
\begin{align}
    \rho_{em} = \rho_{mg} = \frac{e^{-\mu}}{8 \pi r^2}F_0 b^2 \big(\partial_j \xi\big)^2. \label{eq:rhoemg}
\end{align}
Therefore, Eq.~\eqref{eq:stateq} implies in $(\rho_m + 2\rho_{mg})\mathcal{F} + \alpha \Phi \rho_e = 0$. Now we show that such a relation is indeed satisfied. In fact, by using Eqs.~\eqref{eq:rhomg},~\eqref{eq:rhom}, and  \eqref{eq:rhoe} we obtain
\begin{align}
  & \big(\rho_m + 2\rho_{mg}\big)\mathcal{F}_0 + \alpha \Phi_0 \rho_e = \nonumber \\  & =\frac{e^{\mu}}{8 \pi r^2} \Big[\Big(\big(a^2\! - 2 F_0 b^2\big) + 2F_0 b^2\Big)\mathcal{F}_0 + 2\alpha a b \Phi_0 \mathcal{F}_0 \Big]\big(\partial_j \xi\big)^2\nonumber \\
    & = \frac{e^{\mu}}{8 \pi r^2}\big(a + 2\alpha b \Phi_0\big)a\mathcal{F}_0 \big(\partial_j \xi\big)^2 = 0, 
\end{align}
where we have used the relation $\alpha= -a/(2b\,\Phi_0)$, cf. Eq. \eqref{eq:constraint}. 
In conclusion, we find that  the class of solutions obtained by Islam \cite{Islam1977} satisfies Theorem \ref{teo:12}.

Finally, from Eqs.~(\ref{eq:rhom}) and (\ref{eq:rhoe}) it follows a linear relation between $\rho_e$ and $\rho_m$. It is given by
\begin{align}
    {\rho_e} = \frac{2ab \mathcal{F}_0}{a^2 - 2F_0 b^2}\,{\rho_m},
\end{align}
which can assume any finite arbitrary value. The inequality $\rho_m>0$ is guaranteed by appropriate choices of the free parameters $a$, $b$, $F_0$, and ${\cal F}_0 = \pm \sqrt{F_0}$. In particular, we get $\rho_m=\epsilon\rho_e$ for $a=\epsilon b{\cal F}_0\left(1+\sqrt{3}\right)$.

\subsection{Nonzero-pressure rotating charged fluids in the Einstein-Maxwell theory}
\label{sec:3c}

\subsubsection{Differentially rotating charged pressure fluids in the Einstein-Maxwell theory}

Here we study the general properties of differentially rotating charged fluids with nonzero pressure in the Einstein-Maxwell theory. In this general case, the equilibrium equation~(\ref{eq:relequi}), written in terms of total derivatives, reads
\begin{align}
    & (\rho_m + p) d \mathcal{F} + \mathcal{F} d p + \rho_e d \Phi \nonumber  \\ & +  \Big((\rho_m + p )\mathcal{F}^{-1}K -  \rho_e \psi \Big)d \Omega = 0, \label{eq:relequit},
\end{align}
which implies that the four quantities $\mathcal{F}$, $p$, $\Phi$ and $\Omega$ are functionally related. As in the newtonian case, this means that, for instance, the pressure is a function of the other three quantities, i.e., $p = p(\mathcal{F}, \Phi, \Omega)$ which, in turn, implies that $(\partial p/ \partial \mathcal{F})_{\Phi, \Omega} = - (\rho_m + p)/ \mathcal{F}$, $(\partial p/ \partial \Phi)_{\mathcal{F}, \Omega} = - \rho_e/ \mathcal{F}$ and $(\partial p/ \partial \Omega)_{\mathcal{F}, \Phi} = - \Big((\rho_m + p )\mathcal{F}^{-1}K -  \rho_e \psi \Big)/ \mathcal{F}$. Therewith, we can state the following theorem, whose static version can be found in \cite{Guilfoyle:1999yb}. .

\begin{teo}[\textit{differentially rotating and axisymmetric version of Guilfoyle 1999}] \label{teo:4reldif}
For any differentially rotating charged pressure fluid in the Einstein-Maxwell theory with $(\rho_m + p )\mathcal{F}^{-1}K -  \rho_e \psi \neq 0 $, if any three out of the four surfaces of constant $\mathcal{F}$, $\Phi$, $\Omega$, and $p$ coincide, then the fourth also coincides with the other three.
\end{teo}

\begin{proof}
The proof of this theorem is similar to that for Theorem \ref{teo:rel2}. The quantities $\mathcal{F}$, $\Phi$, $\Omega$ and $p$ are scalar fields in spacetime and the conditions of constant $\mathcal{F}$, $\Phi$, $\Omega$, $p$ define two-dimensional surfaces. Moreover, since Eq.~(\ref{eq:relequit}) is valid in its full extent, it is straightforward to see that if any three of the differentials $d\mathcal{F}, d\Phi$, $d \Omega$, and $d p$ vanish, then all of them vanish and the theorem follows.
\end{proof}

This theorem generalizes Theorem \ref{teo:rel2} by including the fluid pressure.

\subsubsection{Differentially rotating charged pressure fluids of Weyl-type in the Einstein-Maxwell theory}
Let us then consider the Weyl ansatz $\mathcal{F} = \mathcal{F}(\Phi)$ and explore some of its consequences for a differentially rotating charged fluid with pressure.  The equilibrium equation \eqref{eq:relequit} now yields
\begin{align}
      &   \left[  \big(\rho_m + p\big) \mathcal{F}' + \rho_e\right] d \Phi + \mathcal{F} d p \nonumber  \\
     &   \qquad  \qquad   \qquad +  \left[\big (\rho_m + p \big)\mathcal{F}^{-1}K -  \rho_e \psi \right] d \Omega = 0, \label{eq:relequi1}
\end{align}
where the prime stands for the derivative with respect to $\Phi$. Therefore, considering $ (\rho_m + p) \mathcal{F}' + \rho_e\neq 0$ and $(\rho_m + p )\mathcal{F}^{-1}K -  \rho_e \psi \neq 0$,  the three quantities $\Phi$, $p$ and $\Omega$ are functionally related and, for instance, the  pressure of the fluid can be considered as a function of $\Phi$ and $\Omega$, i.e., $p = p(\Phi, \Omega)$. This allows us to state two new theorems, which are the general relativistic versions of Theorems~\ref{teo:newt} and \ref{teo:newt1}.  

\begin{teo}[\textit{new}]
If a differentially rotating charged pressure fluid is of Weyl-type, is in equilibrium in the Einstein-Maxwell theory, and obeys the constraints $(\rho_m + p) \mathcal{F}' + \rho_e \neq 0$ and $(\rho_m + p )\mathcal{F}^{-1}K -  \rho_e \psi \neq 0$, then if any two of the surfaces of constant $\Phi$, $\Omega$, or $p$ coincide, then the third also coincides with the other two.
\end{teo}
\begin{proof}
The proof is similar to theorems \ref{teo:rel2} and \ref{teo:4reldif}. From Eq.~(\ref{eq:relequi1}), it is straightforward to see that if any two of the differentials $d\Phi$, $d\Omega$, and $dp$ are zero, then the third differential is also zero and the theorem follows. 
\end{proof}

\begin{teo}[\textit{new}] \label{teo:18}
If a differentially rotating charged pressure fluid is of Weyl-type, is in equilibrium in the Einstein-Maxwell theory, and obeys the constraints $(\rho_m + p) \mathcal{F}' + \rho_e = 0$ and $(\rho_m + p )\mathcal{F}^{-1}K -  \rho_e \psi \neq 0$, then the  surfaces of constant angular velocity coincide with the surfaces of constant pressure.
\end{teo}

\begin{proof}
After the hypothesis $\left(\rho_m + p\right) \mathcal{F}' + \rho_e  = 0$,  the equilibrium equation \eqref{eq:relequit} reduces to $\big[\big(\rho_m + p \big)\mathcal{F}^{-1}K -  \rho_e \psi\big] d \Omega + \mathcal{F}dp = 0$. Therefore, it is straightforward to see that if any one of the differentials $d\Omega$ or $d p$ is zero, then both of them are zero meaning that the surfaces of constant $\Omega$ and $p$ coincide.  
\end{proof}

A consequence of the assumptions made in Theorem \ref{teo:18} is that the ratio between the charge density $\rho_e$ and the enthalpy density $\rho_m +p$, $\rho_e/\left(\rho_m +p\right)=-{\cal F}'$, is a function of $\Phi$ alone. This allows us simplifying the system of equations and finding simple relations between the fluid quantities. For instance, if we consider the simplest ansatz between $\mathcal{F}$ and $\Phi$ in which ${\cal F}= -\epsilon\alpha\Phi + \beta$, with constant $\alpha$ and $\beta$, we get the simple relation $\rho_e= \epsilon\alpha \left(\rho_m +p\right)$. This condition may be of interest in finding exact solutions for differentially rotating charged fluids and is a possible line of studies for future work.

\subsubsection{Rigidly rotating charged pressure fluids in the Einstein-Maxwell theory}

Let us now restrict the analysis to rotating charged pressure fluids in rigid rotation. The equilibrium equation (\ref{eq:relequit}) then reads
\begin{align}
    (\rho_m + p) d \mathcal{F} + \mathcal{F} d p + \rho_e d \Phi = 0. \label{eq:relequi1a}
\end{align}
This relation implies that $\mathcal{F}$, $\Phi$, and $p$ are functionally related as, for instance,  $p = p(\mathcal{F}, \Phi)$, with $ \left(\partial p/ \partial \mathcal{F}\right)_{\Phi}=- \left(\rho_m + p\right) / \mathcal{F}$ and $\left(\partial  p/ \partial \Phi\right)_{\mathcal{F}}=-\rho_e/\mathcal{F}$. Therewith, we can state a new  following theorem which is the rigidly rotating version of a theorem stated by Guilfoyle  \cite{Guilfoyle:1999yb} for nonrotating charged pressure fluids.

\begin{teo}[\textit{Kloster and Das 1977, rigidly rotating and axisymmetric version of Guilfoyle 1999}] \label{teo:4a}
For any rigidly rotating charged pressure fluid in the Einstein-Maxwell theory, if any two among the three surfaces of constant $\mathcal{F}$, $\Phi$, and $p$ coincide, then the third surface also coincides.
\end{teo}

\begin{proof}
As already said above,  $\mathcal{F}$, $\Phi$, and $p$ are scalar functions in spacetime that do not depend upon time $t$. Then, the
conditions of constant $\mathcal{F}$, $\Phi$, and $p$ define level surfaces in the spacetime. Moreover, by Eq.~(\ref{eq:relequi1a}), it is straightforward to see that if any two of the three  differentials $d \mathcal{F} , d\Phi$, and $d p$ vanish, then all of them vanish and the theorem follows.
\end{proof}
This theorem generalizes a result obtained by Guilfoyle \cite{Guilfoyle:1999yb} for static charged pressure fluids by including rigid rotation. The fact that in the case of isometric motion the pressure is a function of the potentials, $p= p({\cal F},\Phi)$,  is also noticed in Ref.~\cite{Kloster}.

\subsubsection{Rigidly rotating charged pressure fluids of Weyl-type in the Einstein-Maxwell theory}
\label{sec:relrigid-p}
Now, let us consider the Weyl ansatz $\mathcal{F} = \mathcal{F}(\Phi)$ and see its consequences for a rigidly rotating fluid with nonzero pressure.  The equilibrium equation now reads
\begin{equation}
    \left[\big(\rho_m +p\big) \mathcal{F}' + \rho_e\right] d \Phi + d p = 0. \label{eq:relequi1b}
\end{equation}
First, let us notice that if $(\rho_m +p) \mathcal{F}' + \rho_e = 0$ then the pressure is constant throughout the fluid. On the other hand, if $(\rho_m +p) \mathcal{F}' + \rho_e \neq 0$, the pressure is also a function of $\Phi$, i.e., $p = p(\Phi)$. This allow us to state a new theorem. 

\begin{teo}[\textit{rigidly rotating and axisymmetric version of Guilfoyle 1999}] \label{teo:5a} 

(i) If a rigidly rotating charged pressure fluid is of Weyl-type and is in equilibrium in the Einstein-Maxwell theory, then the equipotential surfaces are also surfaces of constant pressure, and vice versa.

(ii) If a rigidly rotating charged pressure fluid is of Weyl-type and is in equilibrium in the Einstein-Maxwell theory, then either the pressure gradient vanishes at the surface of the fluid, or the surface is an equipotential surface.
\end{teo}

\begin{proof}
The proof of assertions \textit{(i)} and \textit{(ii)} are analogous of the proof of Theorem \ref{teo:5}. To prove \textit{(i)},  we use the Weyl ansatz $\mathcal{F} = \mathcal{F}(\Phi)$ and the equilibrium equation (\ref{eq:relequi1b}). By the Weyl ansatz, the surfaces of constant $\mathcal{F}$ and $\Phi$ coincide, and Theorem \ref{teo:4a} implies that the surface of constant $p$ also coincides. On the other hand, taking cognizance of Eq.~(\ref{eq:relequi1b}), a surface of constant $p$ implies in $d \Phi = 0$, and thus by the Weyl ansatz it also implies that $d \mathcal{F} = 0$. To prove \textit{(ii)}, we note that smooth boundary conditions require the pressure to vanish at the surface of the fluid. Therefore, the surface of the fluid is a surface of constant pressure, hence $d p = \partial_j p dx^j = 0$ which implies that, either the pressure gradient vanishes at the surface, or by assertion \textit{(i)} the surface of the fluid is an equipotential surface.
\end{proof}

Now, turning into the set of field equations for a rigidly rotating charged fluid with pressure and following the same steps as in Sec. \ref{sec:relrigid}, it is possible to get further properties of rigidly rotating charged pressure fluids. For this task, the two field equations that interest us here are the equations for $F$ and $\Phi$, namely,
\begin{align}
        & \nabla^{\dagger 2} F + \frac{1}{D^2} \Big(F \big(\partial_j K\big)^2 - 2 K \partial_j K \partial_j F - L \big(\partial_j F\big)^2 \Big)\nonumber \\
        & = 8 \pi e^{\mu} F \left(\rho_m + 3p + 2\rho_{em}\right), \label{eq:relpot12}\\
        & \nabla^{\dagger 2} \Phi\!  +\! \frac{1}{D^2}\Big[\!\big(K \partial_j F -\! F \partial_j K \big) \partial_j \psi -\! \big(K \partial_j K +\! L \partial_j F \big) \partial_j \Phi \Big] \nonumber \\
        & = - 4\pi e^{\mu} F^{1/2} \rho_e, \label{eq:relpot13}
\end{align}
where the operator $\nabla^{\dagger 2}$ is given by Eq.~\eqref{eq:nabla1}.
Such equations may be recast into more convenient forms (see Appendix \ref{sec:app3}),
\begin{align}
    & \frac{\partial^j}{D}\Big(\frac{D}{F} \partial_j F\Big) + \frac{1}{F^2 D^2}\big(K \partial^j F - F \partial^j K\big)\big(K \partial_j F - F \partial_j K\big) \nonumber\\
    & = 8 \pi \big(\rho_m + 3 p + 2 \rho_{em}\big), \label{eq:relpot14} \\
    & \frac{\partial^j}{D}\Big(\frac{D}{F} \partial_j \Phi\Big) + \frac{1}{F^2 D^2}\big(K \partial^j F - F \partial^j K\big)\big(K \partial_j \Phi + F \partial_j \psi\big)\nonumber\\ 
    & = - 4 \pi F^{-1/2}\rho_e, \label{eq:relpot15}
\end{align}
where $\partial^j = g^{ji} \partial_i = e^{-\mu} \delta^{ji} \partial_i $. 

In addition, from Eq.~\eqref{eq:rhoel}, we find
\begin{align}
     \frac{2 \Phi}{D}\partial^j \Big(\frac{D}{F} \partial_j \Phi \Big) = \frac{\partial^j}{D} \Big(\frac{D}{F} \partial_j \Phi^2 \Big) - 16\pi\, \rho_{el},\label{eq:relpot16}
\end{align}
 where the identity $\frac{1}{D}\partial^j \Big(\frac{D}{F} \partial_j \Phi^2 \Big) = \frac{2 \Phi}{D}\partial^j \Big(\frac{D}{F} \partial_j \Phi \Big) + \frac{2}{F}\partial^j \Phi \partial_j \Phi$ has been used. 
Hence, by adding Eq.~(\ref{eq:relpot15}) multiplied by $-2\epsilon\alpha \gamma $ to Eq.~(\ref{eq:relpot16}) multiplied by $\alpha$, with $\alpha$
and $\gamma$ being constant parameters, and subtracting the result from Eq.~\eqref{eq:relpot14}, it follows that
\begin{align}
    & \nabla_j \Big[\frac{1}{\mathcal{F}} \nabla^j \Big(F - \alpha( -\epsilon \Phi + \gamma)^2  - \beta\Big) \Big]  = \Xi \,+ \label{eq:relpot17} \\ 
    & 8 \pi \Big[ \big( \rho_m  + 3 p  + 2(1 - \alpha) \rho_{el} + 2 \rho_{mg}\big)\mathcal{F}
    +  \epsilon \alpha \left(\epsilon\Phi - \gamma\right)  \rho_e  \Big], \nonumber
\end{align}
where $\beta$ is another constant parameter and we also substituted the identity $\rho_{em} = \rho_{el} + \rho_{mg}$ into Eq.~(\ref{eq:relpot14}). The symbol $\Xi$ appearing in Eq.~\eqref{eq:relpot17} was introduced to shorten notation and stands for 
\begin{align} 
    \Xi = &  \frac{\mathcal{F}}{F^2D^2}\Big[2\alpha\Phi\big(K \nabla F - F \nabla K\big)\cdot\big(K \nabla \Phi + F \nabla \psi\big)\nonumber \\ &  -  \big(K \nabla F - F \nabla K\big)^2\Big]. 
\end{align}
From the above analysis, a theorem that generalizes a result from Ref.~\cite{Lemos:2009} to rigidly rotating charged fluids may be stated. 

\begin{teo}[\textit{rigidly rotating and axisymmetric version of Lemos and Zanchin 2009}] \label{teo:7a}
 (i) In a rigidly rotating charged pressure fluid is in equilibrium in the Einstein-Maxwell theory, and obeyes the constraint $F - \alpha( -\epsilon \Phi + \gamma)^2 - \beta = 0$, where $\alpha$, $\beta$, and $\gamma$ are constants, then it follows that the right-hand side of Eq.~(\ref{eq:relpot17}) vanishes.

 (ii) In a rigidly rotating charged pressure fluid is in equilibrium in the Einstein-Maxwell theory, the right-hand side of Eq.~(\ref{eq:relpot17}) vanishes, and there is a closed surface, with no singularities, holes, or alien matter inside it, where $F - \alpha( -\epsilon \Phi + \gamma)^2 - \beta= 0$, then it follows that
 \begin{align}
     F = \alpha( -\epsilon \Phi + \gamma)^2 + \beta \label{eq:linearweyl1}
 \end{align}
 everywhere in the fluid.
\end{teo}

\begin{proof}
The proof of this theorem is analogous to the Newton-Maxwell case given in Theorem \ref{teo:7}. The proof of assertion \textit{(i)} is straightforward. By hypothesis, if $F - \alpha( -\epsilon \Phi + \gamma)^2 - \beta = 0$ everywhere in the fluid, it follows that the right-hand side of Eq.~\eqref{eq:relpot17} must be zero. 

In order to proof assertion \textit{(ii)}, let us assume that the right-hand side of Eq.~(\ref{eq:relpot17}) vanishes and define $G \equiv F - \alpha( -\epsilon \Phi + \gamma)^2 - \beta$.  Now the proof can be accomplished by following the same steps as in Theorem \ref{teo:3}\textit{(ii)} for the Newtonian case.  Consider a finite volume $\mathbb{V}_S$
contained in the hypersurface $\Sigma_t$ of $t = \text{constant}$ (see Appendix~\ref{sec:app2}). Now, integrating the divergence 
$\nabla_j \left(\frac{G}{\mathcal{F}}  \nabla^j G \right)$ over $\mathbb{V}_S$, we get $
    \int_{\scriptstyle{\mathbb{V}_{\!S}}} \!\! \nabla_j \left(\frac{G}{\mathcal{F}} \nabla^j G \right) d \mathbb{V}  = \int_{\scriptstyle{\mathbb{V}_{\! \sscrst S}}} \!\mathcal{F}^{-1}\big(\nabla G\big)^2 d\mathbb{V}  = \int_{\!S} \left(\mathcal{F}^{-1} G \nabla_j G\right) n^j dS,$ 
with $S$ being the boundary of the finite volume $\mathbb{V}_S$, $n^j$ being the unit vector normal to $S$, and the Gauss theorem has been used.  If there is a closed surface where $G = 0$, then by identifying such a surface with $S$ one finds $\int_{\scriptscriptstyle\mathbb{V}_{\!S}} \mathcal{F}^{-1}(\nabla G)^2 d\mathbb{V}  = 0$, which is satisfied only if $\nabla_j G=0$ in the volume $\mathbb{V}_S$. Therefore, $G$ must be a constant in the region defined by the volume $\mathbb{V}_S$ of the fluid. Without lost of generality, such a constant may be set to zero and the result for $F$ given in Eq.~\eqref{eq:linearweyl1} follows. 
\end{proof}

Theorem \ref{teo:7a} is the rigidly rotating and axisymmetric version of a theorem stated by Lemos and Zanchin \cite{Lemos:2009}. To see this, let us consider the static  limit of Eq.~(\ref{eq:relpot17}). In this case, we have $K = 0$, $\Omega = 0$, $\psi = 0$,  $F=f$, $\Phi = \phi$, and $\rho_{em} = \rho_{el}$, so that that Eq.~(\ref{eq:relpot17}) reduces to
\begin{align}
    & \nabla_j \left[\frac{1}{\mathcal{F}} \nabla^j \left(F - \alpha\big( -\epsilon \Phi + \gamma\big)^2  - \beta\right) \right]      \\ & = 8 \pi \left[\left(\rho_m + 3 p + 2\left(1-\alpha\right) \rho_{em}\right)\mathcal{F} +\alpha\rho_e\big(\Phi- \epsilon \gamma\big)\right], \nonumber
\end{align}
which is identical to Eq.~(72) in \cite{Lemos:2009} in the case of a four-dimensional spacetime, i.e., with $d=4$. The only difference is that the static theorem stated in \cite{Lemos:2009} assumes no spatial symmetry while here we assumed axial symmetry.

\subsubsection{Rigidly rotating charged pressure fluids of Weyl-type: Islam ansatz}

Here we consider the consequences of the Islam ansatz given in Eqs.~\eqref{eq:ansatz} and \eqref{eq:ansatz1} for a rotating charged pressure fluid of Weyl-type.
In such a case, it is easy to see that Eq.~(\ref{eq:relpot17}) with $\gamma = 0$ yields
\begin{align} 
    & \nabla_j \left[\frac{1}{\mathcal{F}} \nabla^j \left(F - \alpha\,\Phi^2  - \beta\right) \right] =  \label{eq:relpot18} \nonumber \\ &  8 \pi  \left[\left(\rho_m + 3 p +  2\left(1 - \alpha\right)\rho_{el} + 2 \rho_{mg}\right)\mathcal{F} + \alpha \Phi \rho_e\right].
\end{align}
Now, the mentioned ansatz  $\partial_j F = 2\alpha \Phi \partial_j \Phi$ implies in $F = \alpha \Phi^2 + \beta$, and then  we have a new theorem.

\begin{teo}[\textit{new}]\label{teo:22}
If a rigidly rotating charged pressure fluid is in equilibrium in the Einstein-Maxwell theory and satisfies the Islam ansatz, Eqs.~(\ref{eq:ansatz}) and \eqref{eq:ansatz1}, then the fluid quantities obey the relation 
\begin{equation}
\left[\rho_m + 3 p + 2\big(1 - \alpha\big) \rho_{el} + 2 \rho_{mg}\right]\mathcal{F} + \alpha \Phi \rho_e = 0.  \label{eq:statef}
\end{equation}
\end{teo}

\begin{proof}
The theorem is a straightforward consequence of the hypotheses and Eq.~(\ref{eq:relpot18}) and we skip the details. 
\end{proof}

Equation \eqref{eq:statef} can be interpreted as an equation of state analogous to the equations obeyed by Weyl-Guilfoyle nonrotating systems (see Ref. \cite{gautreau}). In fact, it is a generalization of a relation found in \cite{Lemos:2009}. 

Let us now investigate some of the consequences of Theorem \ref{teo:22}. First, consider a system with vanishing Lorentz force. Then one has $\partial_j \Phi=0$, which means a constant $\Phi$, a vanishing $\rho_{el}$ and, moreover, Eq.~\eqref{eq:ansatz} implies in a constant metric potential $F$. Therefore, the equilibrium equation \eqref{eq:relequi1a} implies that $p$ is a constant and Eq.~(\ref{eq:statef}) reduces to $\rho_e = {\rm const.}\times  \left(\rho_m + 3p +  2 \rho_{mg}\right) $. This condition is a consequence of the Islam ansatz and we name it the first Islam-Guilfoyle condition. 

Lastly, for rigidly rotating charged pressure fluids with a nonvanishing Lorentz force, but satisfying the Islam ansatz, we may choose $\mathcal{F} = -\epsilon \sqrt{\alpha}\Phi$, and an alternative form of Eq.~\eqref{eq:statef} turns out, namely, $\rho_m + 3p + 2(1 - \alpha) \rho_{el} + 2 \rho_{mg} -  \epsilon \sqrt{\alpha} \rho_e = 0$. Additionally, the equilibrium condition \eqref{eq:relequi1a} reduces to $\mathcal{F}p'(\Phi) + \rho_e - \epsilon \sqrt{\alpha}(\rho_m + p) = 0$. This condition is also a consequence of the Islam ansatz and we name it the second Islam-Guilfoyle condition.

\section{Final remarks}
\label{sec:4}
We have performed a systematical study on the general properties of rotating and axisymmetric charged fluids with pressure, including Weyl-type systems, both in the Newton-Maxwell theory and in the Einstein-Maxwell theory.

For the Newton-Maxwell theory, we started by considering rotating charged dust fluids, i.e., zero pressure fluids. By directly analyzing the equilibrium equation, we discussed some new results that take into account equilibrium configurations in differential rotation. Moreover, by analyzing the whole set of field equations, we were able to extend to rigidly rotating and axisymmetric fluids a series of well-known theorems for static (nonrotating) charged dust fluids due to Bonnor and others \cite{Bonnor:1980nw,Lemos:2009}. For rigidly rotating charged dust fluids, it follows that the gravitational effective potential $V$ is a function of the electromagnetic effective potential $\Phi$, i.e., $V = V(\Phi)$. In particular, by assuming that the relation between $V$ and $\Phi$ is of the Majumdar-Papapetrou-type, $V = - \epsilon \Phi + \gamma$, then the fluid quantities and the electromagnetic potentials must obey further constraints in order to assure the validity of the Majumdar-Papapetrou relation. On the other hand, one can consider a more general relation between $V$ and $\Phi$ that is the Weyl-Guilfoyle relation, $V = - \epsilon \beta \Phi + \gamma$, loosening the constraint a little more. 
In the case of rotating charged pressure fluids, the Newton-Maxwell system of equations together with the  Weyl ansatz, $V = V(\Phi)$, implies that the pressure is also a function of $\Phi$ alone, and then all fluid quantities are given in terms of the effective electromagnetic potential.  Moreover, we showed that the gravitational effective potential is given by $V = - \epsilon \beta \Phi + \gamma$ if, and only if, two simple relations between the fluid quantities is satisfied, extending the results obtained in Ref.~\cite{Lemos:2009}.  

In the Einstein-Maxwell theory, we started by obtaining new results for equilibrium configurations of rotating charged dust fluids in differential rotation. 
Then, we extended to rigidly rotating charged fluids well-known theorems due to many authors \cite{Majumdar,Papa,das62,deray68,gautreau,Bonnor:1980nw,Guilfoyle:1999yb,Lemos:2009} that hold for nonrotating charged dust fluids. It is verified that the situation is more subtle and more interesting in the relativistic theory than in the Newton-Maxwell theory. More importantly, we showed that for the rotating charged dust fluid,  in principle, a relation of the Majumdar-Papapetrou form $\mathcal{F} = - \epsilon \Phi + \gamma$, also with the condition $\rho_e = \epsilon \rho_m$, is allowed without imposing any relation between the relativistic centrifugal potential $K$ and the magnetic potential $\psi$, provided that a further restriction between the magnetic energy density and the metric and electromagnetic potentials is obeyed. On the other hand, following the works by Bonnor \cite{Bonnor1980b} and Raychaudhuri \cite{Raycha}, who considered an additional relation between $K$ and $\psi$, we proposed a new ansatz, the Islam ansatz involving the gradient of the metric potentials and the gradient of the electromagnetic potentials. This ansatz allows us to simplify the set of equations and then we succeeded to establish neat constraints between the fluid quantities and the electromagnetic energy density. Finally, we discussed a class of solutions obtained by Islam where the Lorentz force vanishes \cite{Islam1977}. 
For rotating nonzero pressure charged fluids, we have done a similar analysis as for rotating charged dust fluids. 
Differential rotation and rigidly rotating axisymmetric fluid distributions were considered separately. 
Several new results were found in both cases. In particular, we showed that a more general relation between the metric potential $F$ and the electromagnetic potential $\Phi$, given by a Weyl-Guilfoyle-type relation of the form $F = \alpha(- \epsilon \Phi + \gamma)^2 + \beta$, implies an intricate relation constraining the fluid quantities and the metric and electromagnetic potentials. In the static limit, this result recovers a theorem due to Lemos and Zanchin \cite{Lemos:2009}. Lastly, we also considered the consequences of the Islam ansatz for rotating charged pressure fluids, and obtained a restriction between the fluid quantities and the electromagnetic energy density. 

Ultimately, as in the case of nonrotating charged fluids of Weyl-type satisfying the Majumdar-Papapetrou or the Weyl-Guilfoyle relations that have been used recently, with some success, to describe static compact objects, such as compact charged stars, regular black holes, and quasiblack holes \cite{lemoskleberzanchin,Lemos:2006sj, lemoszanchin2008,lemoszanchin2010,lemoszanchin2016,lemoszanchin2017}, a possible path to follow is to search for solutions describing rotating charged compact objects. Even though the task of describing rotating compact objects in general relativity is difficult, and most of the known models are constructed by numerical methods~\cite{stergioulas}, in the case of the charged pressure fluids considered in the present work, there is room for further simplifying hypotheses that may allow us finding exact solutions is some particular cases. A further and more interesting, but challenging, step is to look for exact solutions describing compact objects, 
and exploring all sorts of charged rotating objects that the Einstein-Maxwell theory together with the Majumdar-Papapetrou and the Weyl-Guilfoyle relations, or with the Bonnor-Raychaudhuri and the Islam ansatz allow. 
Additionally, other studies about general properties of compact objects composed by rotating charged pressure fluids may be undertaken. For instance, besides implying some constraints among the local quantities as energy density, pressure, and charge density, the mentioned ansatz imply constraints between the global quantities of a finite rotating body, such as gravitational mass, angular momentum, and total electric charge. Building such relations for rotating charged pressure fluids of Islam-Guilfoyle type, for instance, as was done for rotating charged dust fluids in Ref.~\cite{meinel2015} and for rotating charged quasiblack holes with pressure in Ref.~\cite{Lemos:2009wj}, is a work under development.

\begin{acknowledgments}

M.~L.~W.~B.~is funded by the Funda\c c\~ao de Amparo \`a Pesquisa do Estado de S\~ao Paulo (FAPESP), Brazil, Grant No.~2022/09496-8. V.~T.~Z.~is funded in part by Conselho Nacional de Desenvolvimento Cient\'ifico
e Tecnol\'ogico (CNPq), Brazil, Grant No.~311726/2022-4, and by Funda\c c\~ao de Aperfei\c coamento ao Pessoal de N\'ivel Superior (CAPES), Brazil, Grant No. 88887.310351/2018-00. We thank J.~P.~S.~Lemos for stimulating discussions.

\end{acknowledgments}

\appendix 
\section{Some relevant tensor quantities for a rotating charged fluid in a stationary spacetime}
\label{sec:app1}

We employ a coordinate system of cylindrical type, $(x^0,\, x^1,\, x^2,\, x^3) = (t,\, r, \,z, \,\varphi)$, in which the metric can be written in the form
\begin{align}\label{eq:metricap1}
    ds^2 = -f\, dt^2 + 2k\, dt d\varphi + l\, d \varphi^2 + e^{\mu}\left(dr^2 + dz^2\right) ,
\end{align}
where the metric functions $f$, $k$, $l$, and $\mu$ depend on the spatial coordinates $r$ and $z$ only. 
In principle, the ranges of the coordinates are $-\infty< t< \infty$, $0\leq r< \infty$, $-\infty< z< \infty$, and $0\leq \varphi \leq 2\pi$, with $\varphi=0$ identified to $\varphi=2\pi$.

The nontrivial components of the Ricci tensor that are relevant for the present analysis read
\begin{align}
    & \dfrac{2e^{\mu}}{D} R_{tt} = \partial_j\left(\frac{\partial_j f}{D}\right) + \dfrac{f}{D^3} \left(\partial_j f \partial_j l + \left(\partial_j k\right)^2\right), \label{eq:ricci1} \\
    & \dfrac{2e^{\mu}}{D} R_{t\varphi} = - \partial_j\left(\frac{\partial_j k}{D}\right) - \dfrac{k}{D^3}\left(\partial_j f \partial_j l + \left(\partial_j k\right)^2\right), \label{eq:ricci2} \\ 
    & \dfrac{2e^{\mu}}{D} R_{\varphi \varphi} = - \partial_j\left(\frac{\partial_j l}{D}\right) - \dfrac{l}{D^3} \left(\partial_j f \partial_j l + \left(\partial_j k\right)^2\right). \label{eq:ricci3}
\end{align}
The components $R_{rr}$ and $R_{zz}$, that in the present case is identical to $R_{rr}$, are not necessary for our purpose here.

We write down here also the relevant nontrivial components of the energy-momentum tensors $M_{\mu \nu}$ and $E_{\mu \nu}$, together with the expression for the electromagnetic energy density. First we display the general form of fluid velocity $u^\mu$ and of the gauge potential $A^\mu$ adapted to the metric \eqref{eq:metricap1}, namely,
\begin{align}
       u^{\mu} =&\, {\cal F}^{-1}\big(\delta^{\mu}_{\ t} + \Omega \delta^{\mu}_{\ \varphi}\big), \\
    A_{\mu} =&\, \phi \delta_{\mu}^{\ t} + \psi \delta_{\mu}^{\ \varphi},
\end{align}
where 
\begin{equation} 
\mathcal{F}^2  \equiv  f - 2k\Omega - l \Omega^2, \label{eq:calFap}
\end{equation}
and $\Omega$ is the angular velocity of the rotating fluid.
The components of the energy-momentum tensor of the perfect fluid needed for the present analysis are $M_{tt}$, $M_{t\varphi}$, and $M_{\varphi\varphi}$, 
\begin{align}
    & M_{tt} = \mathcal{F}^{-2}\big(f - k \Omega\big)^2 (\rho_m + p)  - f p, \label{eq:Mtt} \\
    & M_{t \varphi} = - \mathcal{F}^{-2}\big(f - k \Omega\big)\big(k + \Omega l\big)\big(\rho_m + p \big) + k p, \label{eq:Mtfi} \\
    & M_{\varphi \varphi} = \mathcal{F}^{-2}\big(k + l \Omega\big)^2 \big(\rho_m + p\big)  - l p. \label{eq:Mfifi}
\end{align}
The $rr$ and the $zz$ components of $M_{\mu\nu}$ are not written because they are not necessary here. On the other hand, the trace $M^{\mu}_{\mu}$ given by
\begin{equation}
    M^{\mu}_{\mu} = - \rho_m + 3 p,\label{eq:traceM}
\end{equation}
is useful.

The relevant components of the electromagnetic energy-momentum tensor $E_{\mu\nu}$ are
\begin{align}
        & E_{t t}  = \frac{e^{-\mu}}{4 \pi D^2 }\left[\left(\frac{1}{2}fl   + k^2\right)\big(\partial_j \phi\big)^2 + fk \partial_j \phi \partial_j \psi\nonumber \right.\\ &\qquad\qquad\qquad\; \left.+ \frac{1}{2}f^2 \big(\partial_j \psi \big)^2 \right], \label{eq:Ett} \\
        & E_{t \varphi} = \frac{e^{-\mu}}{4 \pi D^2 }\left[\frac{1}{2}kl   \left(\partial_j \phi\right)^2 + fl \partial_j \phi \partial_j \psi - \frac{1}{2}f k\big(\partial_j \psi\big)^2 \right], \label{eq:Etfi} \\
        & E_{\varphi \varphi}  = \frac{e^{-\mu}}{4 \pi D^2 }\bigg[\frac{1}{2}l^2\big(\partial_j \phi\big)^2 - kl \partial_j \phi \partial_j \psi \nonumber\\ 
        &\qquad\qquad\qquad\; + \left(\frac{1}{2}fl + k^2\right) (\partial_j \psi)^2 \bigg], \label{eq:Efifi}
\end{align}
where the $rr$ and the $zz$ components of $E_{\mu\nu}$ have been ignored because they are not necessary here.

The expression for the electromagnetic energy density is  built as follows. First we display the expression for the electric field $ E_{\mu}$,
\begin{align}
     E_{\mu} \equiv F_{\mu \nu} u^{\nu} = \delta_{\mu}^{\ j} \mathcal{F}^{-1}(\partial_j \phi + \Omega \partial_j \psi),  
\end{align}
and, second, for the magnetic field $B_{\mu}$,
\begin{align}
        B_{\mu} & \equiv - \frac{1}{2} \epsilon_{\mu \nu \alpha \beta} F^{\alpha \beta}  u^{\nu}\\ & = - \delta_{\mu}^i \epsilon_{t i j \varphi}\frac{\mathcal{F}^{-1}}{D^2}\Big(k \partial^j \phi + f \partial^j\psi + \Omega\big(l \partial^j \phi - k \partial^j \psi\big)\Big).\nonumber  
\end{align}
Therefore, the electric energy density is given by
\begin{align}
     \rho_{el}\equiv \frac{1}{8 \pi} E_{\mu} E^{\mu} = & \, \frac{e^{-\mu}}{8 \pi F} \big(\partial_j \phi + \Omega \partial_j \psi\big)^2, \label{eq:rhoel1} 
\end{align}
while the magnetic energy density reads 
\begin{align}
    \rho_{mg} \equiv &\frac{1}{8 \pi} B_{\mu} B^{\mu}\nonumber \\ 
    =&\, \frac{e^{-\mu}}{8\pi D^2 F}\Big[\big( k + \Omega l \big) \partial_j \phi 
    +\big( f - k\Omega\big) \partial_j \psi\big]^2.
\label{eq:rhomg1}
\end{align}
Finally, the total electromagnetic energy density $\rho_{em}$ is  
\begin{align} \label{eq:rhoema}
    \rho_{em} & = E_{\mu \nu}u^{\mu}u^{\nu} \nonumber\\
    & = \frac{e^{-\mu}}{4 \pi F} \Big[\big(\partial_j \phi + \Omega \partial_j \psi\big)^2 - \frac{F\,l}{2D^2}\big(\partial_j \phi\big)^2 \nonumber \\ & \quad+  \frac{F}{2D^2}\left(2k \partial_j \phi \partial_j \psi + f \big(\partial_j \psi\big)^2 \right) \Big]. 
\end{align}
It is a simple task verifying that the total electromagnetic energy density $\rho_{em}$ is the sum of the electric and magnetic parts, namely,
\begin{align} \label{eq:rhoemb}
 E_{\mu \nu}u^{\mu}u^{\nu}= \frac{1}{8 \pi}\left(  E_{\mu} E^{\mu}+ B_{\mu} B^{\mu}\right) = \rho_{el} + \rho_{mg}.
\end{align}
From the Einstein equations, by properly combining Eqs.~\eqref{eq:ricci1}--\eqref{eq:ricci3} with \eqref{eq:Mtt}--\eqref{eq:Efifi}, it follows that
\begin{align} \label{eq:riccitr}
    R^t_{\ t} + R^{\varphi}_{\ \varphi} = -D^{-1} e^{- \mu}\big(\partial^2_r D + \partial^2_z D\big) = - 16 \pi p. 
\end{align}
This relation implies that, for a dust fluid, i.e., when $p = 0$, $D$ is a harmonic function and therefore one can take (See \cite{Islam2009} for more detail on this subject)
\begin{equation}
    D^2 = f\, l+ k^2 = r^2. \label{eq:D2a}
\end{equation}
 In this case, the expressions for $\nabla^{\dagger 2}$ and $\nabla^{\dagger 2}_-$ given by Eqs.~\eqref{eq:nabla1} and~\eqref{eq:nabla*} become identical to the ones for the Newton-Maxwell theory, i.e., in the Euclidean space, given by Eqs.~\eqref{eq:nabla} and \eqref{nabla**}. 

\section{
Some relevant quantities as seen by an observer comoving with a rigidly rotating fluid}
\label{sec:app2}

By considering a charged fluid in rigid rotation, the coordinate transformation
\begin{align}
    t' = t, \ \ \varphi'  =  \varphi - \Omega t, \ \ r' = r, \ \ z' = z, \label{eq:coordchan}
\end{align}
corresponds to changing to a reference frame comoving (corotating) with the fluid. In the corotating frame, the velocity of the fluid is given by $u'^{\mu} = \mathcal{F}^{-1}\delta^{\mu}_{\ t}$, while the gauge potential transforms to $A'_{\mu} = \Phi \delta_{\mu}^{\ t} + \psi \delta_{\mu}^{\ \varphi}$ with $\Phi = \phi + \Omega \psi$. The coordinate transformation \eqref{eq:coordchan} puts the metric into the form
\begin{align}
 ds^2 = -F dt'^2 + 2K dt'd\varphi' + L d\varphi'^2 + e^{\mu} \big(dr'^2 + dz'^2\big),
\end{align}
where $F = \mathcal{F}^2 = f - 2k \Omega - l \Omega^2, K = k + l \Omega$ and $L =l$. From now on, we can drop the primes.

In turn, the energy density of the electromagnetic field as measured by a comoving observer with the rigidly rotating charged fluid is given by (see Eq.~\eqref{eq:rhoema})
\begin{align}
        \rho_{em}  \equiv & E_{\mu \nu} u^{\mu} u^{\nu} \nonumber = \frac{e^{-\mu}}{4 \pi D^2 F}\bigg[\left(\frac{1}{2}FL   + K^2\right)\left(\partial_j \Phi\right)^2 \nonumber \\ 
      & + FK \partial_j \Phi \partial_j \psi + \frac{1}{2}F^2 (\partial_j \psi)^2 \bigg]. \label{eq:rhoem2}
\end{align}
Similarly, the electric and magnetic parts of $\rho_{em}$ 
are given by (see Eqs.~\eqref{eq:rhoel1} and \eqref{eq:rhomg1})
\begin{align}
    & \rho_{el}\equiv \frac{1}{8 \pi} E_{\mu} E^{\mu} = \frac{e^{-\mu}}{8 \pi F} (\partial_j \Phi)^2, \label{eq:rhoel2} \\
    & \rho_{mg} \equiv \frac{1}{8 \pi} B_{\mu} B^{\mu} = \frac{e^{-\mu}}{8\pi D^2 F}\big(K \partial_j \Phi + F \partial_j \psi\big)^2 , \label{eq:rhomg2} 
\end{align}
respectively.

To use the Gauss integral theorem in the analysis, we need to define an appropriate spacelike volume in the spacetime. Let $\Sigma_t$ be a spacelike hypersurface of constant $x^0=t$, whose unit normal vector is $u^{\mu}$. The induced metric on $\Sigma_t$ is given by $h_{\mu \nu} = g_{\mu \nu} + u_{\mu} u_{\nu}$. For a metric in the form \eqref{eq:metricap1}, it reads 
\begin{align}
    h_{\mu \nu}dx^{\mu} dx^{\nu} = F^{-1} D^2 d\varphi^2 + e^{\mu}\big(dr^2 + dz^2\big), \label{eq:indmetric}
\end{align}
where $D^2= FL+K^2$.
With this, the region $\mathbb{V}_{\sscrst\!S}$ that appears in Theorems~\ref{teo:11} and \ref{teo:7a} is contained in $\Sigma_t$, for each given $t = \text{constant}$. The corresponding volume element in $\mathbb{V}_{\sscrst\!S}$ is well defined and given by $d \mathbb{V} = \sqrt{h} \ d^3 x$, where $h$ is the determinant of the induced metric $h_{\mu \nu}$, namely, $h= F^{-1}D^2\, e^{2\mu}$.

Besides, the covariant divergent of a vector field $X^{\mu}$ on the hypersurface $\Sigma_t$ is given by
\begin{align}
    \nabla_\mu X^\mu \equiv \frac{1}{\sqrt{h}} \partial_\mu \left(\sqrt{h} X^\mu\right).
\end{align}
In particular, for a stationary and axisymmetric vector field given by $X^\mu = \nabla^\mu \Psi$, where $\Psi = \Psi(r,z)$ is an arbitrary scalar field, it follows that
\begin{align}
    \nabla^2 \Psi & \equiv \nabla_\mu \nabla^\mu \Psi = \frac{1}{\sqrt{h}} \partial_\mu\left( \sqrt{h} h^{\mu \nu}\partial_\nu \Psi\right) \nonumber \\
    & =  \frac{e^{- \mu}  \mathcal{F}}{D}  \left[\partial_r \left( \frac{D}{\mathcal{F}} \partial_r \Psi\right) +  \partial_z \left(\frac{D}{\mathcal{F}} \partial_z \Psi\right) \right] + \frac{F}{D^2} \partial_{\varphi}^2 \Psi \nonumber \\
    & = e^{- \mu} \Big( \nabla^{\dagger 2} \Psi - \frac{1}{2} F^{-1} \partial_j F \partial_j \Psi \Big), \label{eq:indlaplc}
\end{align}
where $ \nabla^{\dagger 2}$ is the operator defined in Eq.~\eqref{eq:nabla1}. Therefore, only the derivatives with respect to the coordinates $x^1 = r$ and $x^2 = z$ survive, and  we can write $\nabla^2 \Psi = \nabla_j \nabla^j \Psi$.

\section{On the derivation of some fundamental equations}
\label{sec:app3}

Here we comment on the derivation of Eqs.~\eqref{eq:relpot9c}, \eqref{eq:relpot10c}, \eqref{eq:relpot14}, and \eqref{eq:relpot15} as a sample of the consequences of the coordinate transformation to the corotating frame as made in Secs.~\ref{sec:relrigid} and \ref{sec:relrigid-p}.
Such equations are related to the $tt$, $t\varphi$, and $\varphi\varphi$ components of the Ricci and energy-momentum tensors that enter the field equations. Denoting by $R'_{\mu \nu}$ and  $T'_{\mu \nu}$ the components in the new (primed) frame, and by $R_{\mu \nu}$ and  $T_{\mu \nu}$ the respective components in the original (nonprimed) frame, 
coordinate transformation~\eqref{eq:coordchan} implies
\begin{align}
    & R'_{tt} = R_{tt} + 2 \Omega R_{t \varphi} + \Omega^2 R_{\varphi \varphi}, \\
    & T'_{tt} = T_{tt} + 2 \Omega T_{t \varphi} + \Omega^2 T_{\varphi \varphi}. 
\end{align}
Therefore, in the prime coordinate system, we may obtain an equation for the function $F$ by combining Eqs.~\eqref{eq:Einst1}, \eqref{eq:Einst2}, and \eqref{eq:Einst3}, what is done next.

\subsection{Derivation of Eqs.\texorpdfstring{~\eqref{eq:relpot9c}} and \texorpdfstring{\eqref{eq:relpot10c}} for rigidly rotating charged dust fluids }

In the case of dust fluids, the combination of Eqs.~\eqref{eq:Einst1}, \eqref{eq:Einst2}, and \eqref{eq:Einst3} gives
\begin{align}
    & \nabla^{\dagger 2}_{-} F + \frac{F}{r^2}\Big(\partial_j F \partial_j L + \big(\partial_j K\big)^2 \Big) = 8 \pi e^{\mu} F \big( \rho_m + 2 \rho_{em}\big). \label{eq:Einst123} 
\end{align}
Now, we can interchange $\nabla^{\dagger 2}_{-}$ with $\nabla^{\dagger 2}$ by recalling that, for $p = 0$, one has $D^2 = FL+K^2=r^2$, which implies in
\begin{align}
    \partial_j L = F^{-1}\left(2r \delta_{j}^{\ r} - 2 K \partial_j K - L \partial_j F\right), \label{eq:nablas}
\end{align}
and we get 
\begin{align}
\nabla^{\dagger 2}_{-} = &\nabla^{\dagger 2} - \frac{2\partial_j D}{D}\partial_j \label{eq:nablas2} \\& = \nabla^{\dagger 2} - \frac{1}{D^2}\big(F\partial_j L + L\partial_jF + 2K\partial_j K\big)\partial_j, \nonumber
\end{align}
where $D^2=r^2$ and summation over the Roman index $j$ is assumed. 
Therefore, by substituting \eqref{eq:nablas2} into Eq~\eqref{eq:Einst123}, it follows that
\begin{align}
& \nabla^{\dagger 2} F + \frac{1}{r^2} \Big(F\big(\partial_j K\big)^2 - 2 K \partial_j K \partial_j F - L (\partial_j F)^2 \Big) \nonumber \\ &\qquad\; = 8 \pi e^{\mu} F (\rho_m + 2\rho_{em}),  \label{eq:relpot9}
\end{align}
from which we obtain the desired form, cf. Eq.~\eqref{eq:relpot9c},
\begin{align}
      & \nabla^{2} F + \frac{1}{r^2F} \Big(F \nabla K -  K \nabla F \Big)^2 -  \frac{1}{2 F} (\nabla F)^2 \nonumber \\ &\qquad\; = 8 \pi F \big(\rho_m + 2\rho_{em}\big), 
\end{align}
where $\nabla^2$ is the operator defined in Eq.~\eqref{eq:indlaplc} in the hypersurfaces $\Sigma_t$ of constant $t $ and compatible with the induced metric on $\Sigma_t$, as discussed in Appendix \ref{sec:app2}. Besides,  we used the identities $(\nabla F)^2 \equiv \nabla_j F \nabla^j F = e^{-\mu} (\partial_j F)^2$.

After the coordinate transformation~\eqref{eq:coordchan}, 
the Maxwell equations \eqref{eq:Maxa} and \eqref{eq:Maxb} take the form
\begin{align}
    & F \nabla^{\dagger 2}_{-} \psi + K \nabla^{\dagger 2}_{-} \Phi + \partial_j F  \partial_j \psi +   \partial_j K \partial_j \Phi = 0,  \label{eq:relpot7} \\
    & K \nabla^{\dagger 2}_{-} \psi - L \nabla^{\dagger 2}_{-} \Phi +  \partial_j K \partial_j \psi - \partial_j L \partial_j \Phi = 4\pi r^2 e^{\mu} \rho_e \mathcal{F}^{-1}, \label{eq:relpot8}
\end{align} 
By multiplying Eq.~(\ref{eq:relpot7}) by $K$, Eq.~(\ref{eq:relpot8}) by $F$, and adding the resulting relations we find an equation for the important field $\Phi$,
\begin{align}
        & \nabla^{\dagger 2} \Phi  + \frac{1}{r^2}\!\left[\big(K \partial_j F\!- F \partial_j K \big) \partial_j \psi- \big(K \partial_j K + L \partial_j F \big) \partial_j \Phi \right]\! \nonumber \\ &\qquad\; = - 4\pi e^{\mu} F^{1/2} \rho_e, \label{eq:relpot10}
\end{align}
where we replaced $\nabla^{\dagger 2}_{-}$ with $\nabla^{\dagger 2}$ through convenient manipulation and by using relation \eqref{eq:nablas}. Equation~\eqref{eq:relpot10} can yet be recast in a covariant form, i.e., in the form \eqref{eq:relpot10c},
\begin{align}
    & \nabla^{2} \Phi  + \frac{1}{r^2 F}\big(K \nabla F\!- F \nabla K \big) \cdot ( F \nabla \psi + K \nabla \Phi )  \nonumber \\
    &\qquad\; - \frac{1}{2 F} \nabla F \cdot  \nabla \Phi = - 4\pi F^{1/2} \rho_e,
\end{align}
where the dots indicate indices contraction, e.g., $\nabla F \cdot  \nabla \Phi \equiv \nabla_j F \nabla^j \Phi = e^{-\mu} \partial_j F \partial_j \Phi$.

\subsection{Derivation of Eqs.\texorpdfstring{~\eqref{eq:relpot14}} and \texorpdfstring{\eqref{eq:relpot15}} for the rigidly rotating charge pressure fluid}

In the case of a nonzero pressure fluid, the combination of Eqs.~\eqref{eq:Einst1}, \eqref{eq:Einst2}, and \eqref{eq:Einst3} gives
\begin{align}
     \nabla^{\dagger 2}_{-} F + \frac{F}{D^2} & \Big(\partial_j F \partial_j L + \big(\partial_j K\big)^2 \Big) \nonumber  \\ & = 8 \pi e^{\mu} F \big( \rho_m + 3p + 2 \rho_{em}\big). \label{eq:Einst123p}
\end{align}
Now,  by substituting relation \eqref{eq:nablas2} into Eq.~\eqref{eq:Einst123p}, the operator $\nabla^{\dagger 2}_{-}$ is eliminated and we get Eq.~\eqref{eq:relpot12}. The same steps can be followed in order to obtain Eq.~\eqref{eq:relpot13}.

However, Eq.~(\ref{eq:relpot12}) can yet be recast in a more convenient form. By noticing that
\begin{align}
    & \frac{1}{D} \partial_j \Big(\frac{D}{F}\partial_j F \Big) = \frac{1}{F} \nabla^{\dagger 2} F - \frac{1}{F^2} \big(\partial_j F\big)^2, \\
    & \frac{L}{D^2} = \frac{D^2 - K^2}{F D^2},
\end{align}
the left-hand side of Eq.~\eqref{eq:relpot12} can be rewritten as
\begin{align}
         & \frac{1}{F} \nabla^{\dagger 2} F + \frac{1}{F D^2} \Big(F \big(\partial_j K\big)^2 - 2 K \partial_j K \partial_j F - L \big(\partial_j F\big)^2 \Big) \nonumber \\ & = \frac{1}{D} \partial_j \Big(\frac{D}{F}\partial_j F \Big) + \frac{1}{F^2 D^2}\big(K \partial_j F - F \partial_j K\big)^2.             
\end{align}
After this, it is seen that Eq.~\eqref{eq:relpot14} immediately follows from Eq.~\eqref{eq:relpot12}. Finally, by following the same steps, Eq.~\eqref{eq:relpot15} may be derived from Eq.~\eqref{eq:relpot13}.


\begin{thebibliography}{10}
\bibliographystyle{apsrev4-2}


\bibitem{Weyl} H. Weyl, Zur Gravitationstheorie, Ann. Phys. (Berlin) {\bf 54}, 117 (1917).
English translation (Republication):
H. Weyl, On the theory of gravitation. Gen. Relativ. Gravit. {\bf 44}, 779 (2012). 

\bibitem{Majumdar} S. D. Majumdar, A class of exact solutions of Einstein's field equations, Phys. Rev. {\bf 72}, 390 (1947).

\bibitem{Papa} A. Papapetrou, A static solution of the equations of the gravitational field for an arbitrary charge-distribution, Proc. R. Irish Acad. A {\bf 51}, 191 (1947).

\bibitem{Bonnor:1953ex} W. B. Bonnor, Certain exact solutions of the equations of general relativity with an electrostatic field, Proc. Phys. Soc. A {\bf 66}, 145 (1953).

\bibitem{das62} A. Das, A class of exact solutions of certain classical field equations in general relativity, Proc. R. Soc. A {\bf 267}, 1 (1962).

\bibitem{bonnor65} W. B. Bonnor, The equilibrium of a charged sphere, Mon. Not. R. Astron. Soc. {\bf 129}, 443 (1965).

\bibitem{deray68} U. K. De and A. K. Raychaudhuri, Static distribution of charged dust in general relativity, Proc. R. Soc. A {\bf 303}, 47 (1968).

\bibitem{de68} U. K. De, Non-static spherically symmetric charged dust distribution in general relativity, J. Phys. A {\bf 1}, 645 (1968).

\bibitem{Bonnor:1972wi} W. B. Bonnor and S. B. P. Wickramasuriya, A static body of arbitrarily large density, Int. J. Theor. Phys. {\bf 5}, 371 (1972).

\bibitem{gautreau} R. Gautreau and R. B. Hoffman, The structure of the sources of weyl-type electrovac fields in general relativity, Nuovo Cimento B {\bf 16}, 162 (1973).

\bibitem{bonnor75} W. B. Bonnor and S. B. P. Wickramasuriya, Are very large gravitational redshifts possible?, Mon. Not. R. Astron. Soc. {\bf 170}, 643 (1975).

\bibitem{Bonnor:1980nw} W. B. Bonnor, Equilibrium of charged dust in general relativity, Gen. Relativ. Gravit. {\bf 12}, 453 (1980).

\bibitem{Gurses:1998zu}  M. G\"urses, Sources for the Majumdar-Papapetrou space-times, Phys. Rev. D {\bf 58}, 044001 (1998).

\bibitem{ida} D. Ida, Static charged perfect fluid with the Weyl-Majumdar relation, Prog. Theor. Phys. {\bf 103}, 573 (2000).

\bibitem{Ivanov:2002jy} B. V. Ivanov, Static charged perfect fluid spheres in general relativity, Phys. Rev. D {\bf 65}, 104001 (2002).

\bibitem{varela}  V. Varela, Construction of sources for Majumdar-Papapetrou spacetimes, Gen. Relativ. Gravit. {\bf 35}, 1815 (2003).

\bibitem{Lemos:2009} J. P. S. Lemos and V. T. Zanchin, Electrically charged fluids with pressure in Newtonian gravitation and general relativity in $d$ spacetime dimensions: Theorems and results for Weyl-type systems, Phys. Rev. D {\bf 80}, 024010 (2009).

\bibitem{Guilfoyle:1999yb} B. S. Guilfoyle, Interior Weyl-type solutions of the Einstein-Maxwell field equations, Gen. Relativ. Gravit. {\bf 31}, 1645 (1999)

\bibitem{bonnor99}  W. B. Bonnor, Comment on 'Relativistic charged spheres: II. Regularity and stability', Classical Quantum Gravity {\bf 16}, 4125 (1999).

\bibitem{Lemos:2003gx}  J. P. S. Lemos and E. J. Weinberg, Quasi-black holes from extremal charged dust, Phys. Rev. D {\bf 69}, 104004 (2004).

\bibitem{lemoskleberzanchin} A. Kleber, J. P. S. Lemos, and V. T. Zanchin, Thick shells and stars in Majumdar-Papapetrou general relativity, Gravitation Cosmol. {\bf 11}, 269 (2005).

\bibitem{Lemos:2006sj} J. P. S. Lemos and V. T. Zanchin, Gravitational magnetic monopoles and Majumdar-Papapetrou stars, J. Math. Phys. {\bf 47}, 042504 (2006).

\bibitem{lemoszanchin2008} J. P. S. Lemos and V. T. Zanchin, Bonnor stars in $d$ spacetime dimensions, Phys. Rev. D {\bf 77}, 064003 (2008).

\bibitem{lemoszanchin2010} J. P. S. Lemos and V. T. Zanchin, Quasiblack holes with pressure: Relativistic charged spheres as the frozen stars, Phys. Rev. D {\bf 81}, 124016 (2010).

\bibitem{lemoszanchin2016} J. P. S. Lemos and V. T. Zanchin, Regular black holes: Guilfoyle's electrically charged solutions with a perfect fluid phantom core, Phys. Rev. D {\bf 93}, 124012 (2016).

\bibitem{lemoszanchin2017} J. P. S. Lemos and V. T. Zanchin, Plethora of relativistic charged spheres: The full spectrum of Guilfoyle's static, electrically charged spherical solutions, Phys. Rev. D {\bf 95}, 104040 (2017).

\bibitem{masalemoszanchin2023} A.~D.~D.~Masa, J.~P.~S.~Lemos, and V.~T.~Zanchin,
Stability of electrically charged stars, regular black holes, quasiblack holes, and quasinonblack holes,
Phys. Rev. D \textbf{107}, 064053 (2023).

\bibitem{IsraelWilson} W. Israel and G. A. Wilson, A Class of stationary electromagnetic vacuum fields, J. Math. Phys. {\bf 13}, 865 (1972).

\bibitem{Perjes1971} Z. Perjes, Solutions of the Coupled Electromagnetic Equations Representing the Fields of Spinning Sources, Phys. Rev. Lett. {\bf 27}, 1668 (1971).

\bibitem{Ernst1} F. J. Ernst, New formulation of the axially symmetric gravitational field problem, Phys. Rev. {\bf 168}, 1415 (1968).

\bibitem{IsraelSpanos1973}
W. Israel, J.~T.~J. Spanos, Equilibrium of charged spinning masses in general relativity, Lett. Nuovo Cimento {\bf 7}, 245 (1973). 

\bibitem{Gurses:2023csy}
M.~Gurses, T.~C.~Sisman and B.~Tekin,
Israel-Wilson-Perjes metrics in a theory with a dilaton field,
Phys. Rev. D \textbf{108},  024060 (2023).

\bibitem{Islam1978} J. N. Islam, On rotating charged dust in general relativity, Proc. R. Soc. A {\bf 362}, 329 (1978). 

\bibitem{Bonnor1980b} W. B. Bonnor, Rotating charged dust in general relativity, J. Phys. A. {\bf 13}, 3465 (1980).
\bibitem{Islam1979} J. N. Islam, On rotating charged dust in general relativity. II, Proc. R. Soc. A {\bf 367}, 271 (1979).

\bibitem{Islam1980} J. N. Islam, On rotating charged dust in general relativity. III, Proc. R. Soc. A {\bf 372}, 111 (1980).

\bibitem{Islam1983} J. N. Islam, On rotating charged dust in general relativity. IV, Proc. R. Soc. A {\bf 385}, 189 (1983).

\bibitem{Islam1983a} J. N. Islam, On rotating charged dust in general relativity. V, Proc. R. Soc. A {\bf 389}, 291 (1983).

\bibitem{Islam2009} J. N. Islam, \textit{Rotating Fields in General Relativity} (Cambridge University Press, 1985).

\bibitem{Islam1977} J. N. Islam, A class of exact interior solutions of the Einstein-Maxwell equations, Proc. R. Soc. A {\bf 353}, 523 (1977).

\bibitem{ChakraBand1} S. K. Chakraborty and N. Bandyopadhyay, Rigidly rotating charged dust distribution in general relativity, J. Phys. A. {\bf 16}, 3013 (1983).

\bibitem{Raycha} A. K. Raychaudhuri, Rotating charged dust in general relativity, J. Phys. A. {\bf 15}, 831 (1982).

\bibitem{Som} M. M. Som and  A. K. Raychaudhuri, Cylindrically symmetric charged dust distributions in rigid rotation in general relativity, Proc. R. Soc. A {\bf 304}, 81 (1968).

\bibitem{Banerjee} A. Banerjee, N. Chakravarty, and S. B. Dutta Choudhury, On stationary distributions of charged dust, Aust. J. Phys. {\bf 29} 119 (1976).

\bibitem{Misra} R. M. Misra, D. B. Pandey, and D. C. Srivastava, Stationary charged-dust distribution in general relativity, Nuovo Cimento B {\bf 9}, 64 (1972). 

\bibitem{Bergh1983} P. Wils and N. Van den Bergh, A class of stationary Einstein-Maxwell solutions with cylindrical symmetry, J. Phys. A. {\bf 16}, 3843 (1983).

\bibitem{Bergh1984} P. Wils and N. Van den Bergh, Rotating charged dust as a source for cylindrically symmetric electrovacs, Proc. R. Soc. A  {\bf 394}, 437 (1984).

\bibitem{Bergh} N. Van den Bergh and P. Wils, Differentially rotating charged dust with a force-free electromagnetic field, Classical Quantum  Gravity {\bf 1}, 199 (1984).

\bibitem{Wils}  N. Van den Bergh and P. Wils, Differentially rotating charged dust with a force-free electromagnetic field II, Classical Quantum  Gravity {\bf 1}, 399 (1984).

\bibitem{Islam1984} J. N. Islam, P. Wils, and N. Van den Bergh, General solutions for axisymmetric differentially rotating charged dust with vanishing Lorentz force, Classical Quantum  Gravity {\bf 1}, 705 (1984).

\bibitem{Wils1985} P. Wils and N. Van den Bergh, The rotation axis for stationary and axisymmetric space-times, Classical Quantum  Gravity {\bf 2}, 229 (1985).


\bibitem{ChakraBand2}  S. K. Chakraborty and N. Bandyopadhyaya, Charged perfect fluid in rigid rotation, J. Math. Phys. {\bf 26}, 1752 (1985).

\bibitem{Das1977} A. Das and S. Kloster,  Stationary gravitational fields of a dually charged perfect fluid , J. Math. Phys. {\bf 18}, 586 (1977).

\bibitem{Kloster} S. Kloster and A. Das, Stationary gravitational fields of a charged perfect fluid, J. Math. Phys. {\bf 18}, 2191 (1977).

\bibitem{Khater} A. H. Khater and M. J. Mourad, Rotating ideal magnetohydrodynamic flow in general relativity, Astrophys. Space. Sci. {\bf 152}, 215 (1989).

\bibitem{Goldreich}  P. Goldreich and W. H. Julian, Pulsar electrodynamics, Astrophys. J. {\bf 157}, 869 (1969).

\bibitem{Shapiro} S. L. Shapiro, S. A. Teukolsky, \textit{Black Holes, White Dwarfs and Neutron Stars: The Physics of Compact Objects} (Wiley-Vch, 1991).

\bibitem{thorne} K.S. Thorne, \textit{Relativistic Stars, Black Holes and Gravitational Waves (Including an in Depth Review of the Theory of Rotating, Relativistic Stars)}, contained in B. K. Sachs, \textit{General Relativity and Cosmology} (Academic Press, 1971).

\bibitem{abra} M. A. Abramowicz, The relativistic Von Zeipel's theorem, Acta Astron. {\bf 21}, 81 (1970).

\bibitem{RayDe} A. K. Raychaudhuri and  U. K. De, Charged-dust distributions in general relativity, J. Phys. A: Gen. Phys. {\bf 3}, 263 (1970).

\bibitem{Collins} C. B. Collins, Shear-free fluids in general relativity, Can. J. Physics {\bf 64(2)}, 191 (1986).

\bibitem{meinel2015}  M. Breithaupt, Y.-C. Liu, R. Meinel, and S. Palenta, On the black hole limit of rotating discs of charged dust, Classical Quantum  Gravity {\bf 32}, 135022 (2015). 

\bibitem{Lemos:2009wj}
J.~P.~S.~Lemos and O.~B.~Zaslavskii,
Angular momentum and mass formulas for rotating stationary quasi-black holes,
Phys. Rev. D \textbf{79}, 044020 (2009).

\bibitem{stergioulas} V. Paschalidis and N. Stergioulas, Rotating stars in relativity, Living Rev. Relativity {\bf 20}, 7 (2017).



\end{thebibliography}
\end{document}